\documentclass[12pt]{article}

\usepackage[T1]{fontenc}
\usepackage{amsmath}
\usepackage{amssymb}
\usepackage{amsthm}
\usepackage{mathtools}
\usepackage{url}

\usepackage{subfigure}
\usepackage{graphicx}

\usepackage{bigfoot} 

\usepackage{geometry}
\usepackage{hyperref}
\hypersetup{
	pdftoolbar=true,
	pdfmenubar=true,
	pdffitwindow=false,
	pdfstartview={FitH},
	pdftitle={},
	pdfauthor={Dario Corona},
	pdfsubject={},
	pdfkeywords={},
	pdfnewwindow=true,
	colorlinks=true, 
	linkcolor=blue,
	citecolor=blue,
	urlcolor=blue,
}

  \usepackage[
    backend=biber,
    style=numeric,
  ]{biblatex}

\theoremstyle{plain}\newtheorem{teo}{Theorem}
\theoremstyle{plain}\newtheorem{proposition}[teo]{Proposition}
\theoremstyle{plain}
\theoremstyle{plain}

\theoremstyle{definition}\newtheorem{definition}[teo]{Definition}
\theoremstyle{definition}\newtheorem{assumption}[teo]{Assumption}
\theoremstyle{remark}\newtheorem{remark}[teo]{Remark}
\theoremstyle{definition}\newtheorem{example}[teo]{Example}

\numberwithin{teo}{section} 
\numberwithin{equation}{section}

\usepackage[colorinlistoftodos,prependcaption,textsize=tiny]{todonotes}
\newcounter{mycomment}

\usepackage{authblk}
\usepackage{color}

\DeclareMathOperator{\sgn}{sgn}
\newcommand{\de}{\mathrm{d}}
\title{Global models of collapsing scalar field: endstate}
\author[1]{Dario Corona\thanks{dario.corona@unicam.it}}
\author[2]{Roberto Giamb\`o\thanks{roberto.giambo@unicam.it}}
\affil[1,2]{School of Science and Technology, University of Camerino, Camerino (MC), Italy}
\affil[2]{Istituto Nazionale di Fisica Nucleare, Sez. di Perugia, Perugia, Italy}
\affil[2]{IAPS, Istituto Nazionale di Astrofisica - Tor Vergata, Roma, Italy}

\begin{document}
\maketitle

\abstract{The study of dynamic singularity formation in spacetime, focusing on scalar field collapse models, is analysed.
We revisit key findings regarding open spatial topologies, concentrating on minimal conditions necessary for singularity and apparent horizon formation.
Moreover, we examine the stability of initial data in the dynamical system governed by Einstein's equations, considering variations in parameters that influence naked singularity formation.
We illustrate how these results apply to a family of scalar field models, concluding with a discussion on the concept of genericity in singularity studies.
}

\section{Introduction}
Starting from the middle of the last century, Hawking and Penrose's theorems on geodesic incompleteness and the emergence of black hole candidates sparked interest in dynamic singularity formation.
As interest in spacetime dynamics grew, it became evident that Schwarzschild or Kerr spacetimes could merely represent an asymptotic state, necessitating the consideration of richer solutions to understand relativistic collapse.
Because of their more manageable geometry, researchers turned to spherically symmetric models, particularly focusing on collapse in dust clouds, radiation, and scalar fields. 

Research on the collapse of scalar fields began in the late 1960s, and the mathematical problem was approached in a variety of ways, notably in a renowned series of papers by D. Christodoulou \cite{Christodoulou:1991yfa,https://doi.org/10.1002/cpa.3160460803,Christodoulou:1994hg,ab6b0194-3818-3102-ab03-e8d850e0918d}.
Christodoulou's work laid the foundation for the theories by proving conditions leading to trapped surfaces and addressing instability with respect to initial data, which advanced the understanding of scalar field collapse and cosmic censorship.

Several years later, a new avenue emerged, spurred by connections to extended gravity theories and string theory.
Interest grew in self-interacting models with a potential: indeed, the works cited above dealt at that point only with the free massless case, and refinements where potentials come into play were largely limited to the analysis of non-central singularities \cite{Dafermos:2004ws}.
On the other hand, the new line of research, often within a cosmological context, studies \textit{homogeneous} models with a Robertson--Walker (RW) spacetime background, which simplifies equations into an ODE system. 

A wide array of potentials $V(\phi)$ has been explored, investigating conditions leading to future singularities in a finite amount of RW time.
One of the first and probably most important results in this realm is the work by Scott Foster \cite{Foster:1998sk}, which introduced methods and techniques that influenced all the following advances on this topic.
For further discussion, see 
\cite{Miritzis:2003ym,Miritzis:2003eu,Miritzis:2005hg} and therein references, including those in the present paper. 

At the same time, the picture emerged where, under suitable conditions, the formation of the apparent horizon in this spacetime is such that one can build a  model by gluing a RW interior with an external spherical metric.
In this way, one obtains a global model where, starting from regular initial data, the apparent horizon does not form.
In analogy to the example shown in \cite{Christodoulou:1994hg} for non-homogeneous collapse, this example of naked singularity was examined to understand its \textit{genericity} in terms of tunable parameters.
By genericity we mean here the stability with respect to small perturbations of the parameters leading to a certain feature, e.g., the appearance of a naked singularity in the modality described above.
This research intersected with efforts to identify quantum mechanisms for preventing singularity formation (see, e.g., the review \cite{Malafarina:2017csn}). 

In this paper, we revisit the main results concerning open spatial topologies (specifically, when $\kappa=0$ or $\kappa = -1$)
with the goal of employing the minimal necessary assumptions on the potential function $V(\phi)$.
While our focus is on elucidating the arguments underpinning their proofs, it is important to note that our scope does not encompass the vast literature on the subject.
For instance, \cite{Tzanni:2014eja,Giambo:2014jfa,Giambo:2015tja} delve into potentials $V(\phi)$ taking negative values on non-compact subsets of $\mathbb{R}$.
Conditions on $V(\phi)$ leading to generic singularity formation and to generic formation of the apparent horizon will be derived. 

The genericity results presented here refer to the stability with respect to the choice of the initial data for the dynamical system ruling these models, essentially derived from Einstein field equations where gravity is coupled with the scalar field self--interacting with $V(\phi)$. 
In fact,  one can imagine an analysis of these models where a varying parameter affects the functional dependence of a physical quantity -- e.g. the energy density -- on the scale factor of RW spacetime.
In this way, an alternative and different notion of genericity is obtained, which can lead to a completely different interpretation: in \cite{Goswami:2005fu,Baier:2014ita,Mosani:2021czj} such situations are shown, where naked singularity formation is stable with respect to the choice of the parameter. 
The effect of altering this parameter is to simultaneously modify both the function $V(\phi)$ and the initial data,
thus obtaining a family of scalar field star models collapsing to a naked singularity.
The theory we are discussing naturally applies to each member of this family, whose naked singularity formation is unstable to a perturbation of the initial data if the potential is held fixed. Noticeably enough, however, these naked singularities have been proved with numerical methods to be gravitationally strong \cite{Guo:2020ked}. 

The paper is organized as follows. Section \ref{sec:the_scalar_field_collapsing_model} presents the interior RW metric and the equations ruling its dynamics. The asymptotic behavior of the variables in the open topologies  $\kappa=0$ and $\kappa=-1$ are  studied in Sections \ref{sec:the_flat_case_kappa_0_} and \ref{sec:collapse_of_the_open_spatial_topology_kappa_1_}, respectively.
The construction of the global model is then given in Section \ref{sec:global_model_the_exterior}, together with the main theorems on the genericity of singularity and horizon formation. The different notions of genericity, sketched above, are analyzed in full details in \ref{sec:horizon_formation_genericity}, referencing to the above cited examples from the literature. Section \ref{sec:outro} is devoted to the conclusions.

\section{The scalar field collapsing model}%
\label{sec:the_scalar_field_collapsing_model}
First of all, let us work in normalized units in such a way that the Einstein field equation will read as follows: 
\[
G+\Lambda g=T.
\]
In order to build a global model, we will consider an interior matter
where the energy momentum tensor is given by
\begin{equation}
	\label{eq:momentumTensor}
	T_{\mu\nu} =
	\phi_{,\mu}\phi_{,\nu}
	- \left(
		\frac{1}{2}g^{\alpha\beta}\phi_{,\alpha}\phi_{,\beta}
		+ V(\phi)
	\right)g_{\mu\nu},
\end{equation}
where $V(\phi)$ is a given potential that depends on a scalar field $\phi$, a function defined on the spacetime interior manifold.
We immediately observe that the potential $V(\phi)$ can also embody the contribution of the cosmological constant $\Lambda$, via its redefinition $V\longrightarrow V+\Lambda$.

We suppose that the gradient of this function is timelike,
allowing us to choose a comoving gauge
such that the metric is given by a homogeneous and spherical RW model:
\begin{equation}
	\label{eq:RWmodel}
	g = -\de t^2
	+ a^2(t)\left[\frac{\de r^2}{1 - \kappa r^2} + r^2\de\Omega^2\right],
\end{equation}
where $\de\Omega^2$ is the line element of $S^2$ embedded
in the euclidean space $\mathbb{R}^3$,
$a(t)$ is the scale factor of the homogeneous interior,
and $\kappa$ is a  parameter related to the 
(constant) curvature of the spatial part.
As is well known, relevant cases are given by $\kappa = -1,0,+1$.
The Einstein field equations then become:
\begin{subequations}
	\begin{align}
		G_0^0 = T_0^0\quad \implies \quad
		& -\frac{3(\kappa + \dot{a}^2(t))}{a^2(t)}
		= - \left(\frac{1}{2}\dot{\phi}^2(t)+ V(\phi(t))\right),
		\label{eq:Einstein-0}\\
		G_i^i = T_i^i\quad\implies\quad
		&
		- \frac{\kappa+\dot{a}^2(t) + 2 a(t)\ddot{a}(t)}{a^2(t)}
		= \frac{1}{2}\dot{\phi}^2(t) - V(\phi(t)),
		\label{eq:Einstein-i}
	\end{align}
\end{subequations}
which implies the Bianchi identity $T^{\mu\nu}_{\ \ ;\nu} = 0$,
now expressed as follows:
\begin{equation}
	\label{eq:Bianchi-here}
	\dot{\phi}(t)
	\left(
		\ddot{\phi}(t) + V'(\phi(t))
		+ 3\frac{\dot{a}(t)}{a(t)}\dot{\phi}(t)
	\right) = 0.
\end{equation}
By~\eqref{eq:Einstein-0},
solutions such that $\dot{\phi}(t) \equiv 0$ correspond to an energy momentum tensor given by 
$T^\mu_\nu = - V(\phi_0)\delta^\mu_\nu$,
hence to vacuum solutions with cosmological constant
$\Lambda = V(\phi_0)$,
i.e., to (anti-) de Sitter solutions.
Apart from this particular case, the scalar field $\phi$
is a solution of Klein-Gordon equation:
\begin{equation}
	\label{eq:Klein-Gordon}
	\Box\phi + V'(\phi) = \ddot{\phi}(t) + V'(\phi(t))
	+ 3 \frac{\dot{a}(t)}{a(t)}\dot{\phi}(t) = 0,
\end{equation}
where $\Box$ is the D'Alambert operator induced by $g$.

\begin{remark}\label{rem:Rem1}
	To gain further insight on the interplay
	between~\eqref{eq:Einstein-0}-\eqref{eq:Einstein-i}
	and~\eqref{eq:Klein-Gordon},
	let us introduce the Hubble function
	\begin{equation*}
		h(t) = \frac{\dot{a}(t)}{a(t)},
	\end{equation*}
	and consider the function
	\begin{equation}
		\label{eq:W-def}
		W(t) \coloneqq
		h^2(t) + \frac{\kappa}{a^2(t)}
		- \frac{1}{3}\left(\frac{1}{2}\dot{\phi}^2(t) + V(\phi(t))\right).
	\end{equation}
	Of course, $W(t) = 0$ is precisely~\eqref{eq:Einstein-0}.
	But observe that, using~\eqref{eq:Einstein-i} and~\eqref{eq:Klein-Gordon},
	we have
	\begin{equation}
		\label{eq:KleinGordonODE}
		\dot{W}(t) = -3h(t)W(t),
	\end{equation}
	i.e., $W(t) = W_0a^{-3}(t)$,
	so that, assuming~\eqref{eq:Einstein-i} and~\eqref{eq:Klein-Gordon},
	together with a choice of initial data such that $W_0 = 0$,
	then $W(t) = 0$ for all $t$.
\end{remark}
\begin{remark}
	\label{rem:Rem2}
	Using~\eqref{eq:Klein-Gordon}
	along with~\eqref{eq:Einstein-0}-\eqref{eq:Einstein-i},
	we derive
	\begin{equation}
		\label{eq:h-ODE}
		\dot{h}(t) = \frac{\kappa}{a^2(t)}-\frac{1}{2}\dot{\phi}^2(t),
	\end{equation}
	and therefore in the case of open spatial topologies,
        i.e., $\kappa=0$ and $\kappa=-1$,
	we can conclude that $h(t)$ is a decreasing function
	throughout the evolution.
    This fact motivates the interest, in this paper, for spherical models with $\kappa=0$ or $\kappa=-1$. 
	As a consequence, since we are interested in collapsing solutions,
	we choose the initial state such that $\dot{a}(0) < 0$,
	which implies $h(t) < 0$ for all $t \ge 0$.
\end{remark}

Using the above remark, we will employ a normalization 
scheme widely used in literature, introducing the new variables $x,w$ and $z$ defined as follows:
\begin{equation}
	\label{eq:def-xwz}
	x = -\frac{1}{\sqrt{3}\, h}, \quad
	w = -\frac{\dot{\phi}}{\sqrt{6}\, h},\quad
	z = -\frac{1}{h a} \left(= -\frac{1}{\dot{a}}\right).
\end{equation}
Moreover, we employ a normalized line $\tau$:
\begin{equation}
	\label{eq:def-tau}
	\de \tau = -\sqrt{6}\, h\, \de t.
\end{equation}
Using~\eqref{eq:Einstein-i}, \eqref{eq:Klein-Gordon} and \eqref{eq:h-ODE}
we get the system:
\begin{subequations}
	\begin{align}
		&\frac{\de\phi}{\de\tau} = w,
		\label{eq:mainSys-a}\\
		&\frac{\de x}{\de\tau}
		=- \frac{x}{\sqrt{6}}
		\left(3w^2 -\kappa z^2\right),
		\label{eq:mainSys-b}\\
		&\frac{\de w}{\de\tau}
		= \sqrt{\frac{3}{2}}
		w
			\left( 1 -w^2 + \frac\kappa 3 z^2 \right)
			- \frac12\, V'(\phi)x^2,
		\label{eq:mainSys-c}\\
		&\frac{\de z}{\de\tau}
		= \frac{z}{\sqrt{6}}
		\left(1-3w^2 + \kappa z^2\right).
		\label{eq:mainSys-d}
	\end{align}
\end{subequations}
As a matter of fact, the above system is not free
since, as we have seen in Remark~\ref{rem:Rem2},
we have to take into account the condition~\eqref{eq:Einstein-0},
which is guaranteed if we select initial data satisfying it.
This condition, when applied to the new set of unknowns, takes the following form,
\begin{equation}
\label{eq:Constraint}
V(\phi)x^2 + w^2 - \kappa z^2 = 1,
\end{equation}
which allows $x$ to be decoupled from the other unknowns in the system,
assuming $V(\phi) \ne 0$.
This point will be further explored in subsequent discussions.



If $x(0) > 0$ and $\kappa = 0,-1$,
then~\eqref{eq:mainSys-b} implies 
that $x(\tau)$ is a decreasing and positive function,
hence bounded for every $\tau \ge 0$.
Postulating that $V(\phi)$ is bounded from below,
as as actually stipulated by Assumption~\ref{as:V},
we deduce from~\eqref{eq:Constraint}
that $w$ and $z$ will also be bounded for every $\tau\ge 0$.
The only unknown that can (and indeed does)
attain unbounded values is $\phi(\tau)$.
To control the behaviour 
at infinity of the scalar field,
we introduce a new variable $s$,
related to $\phi$ by 
\begin{equation*}
	s = f(\phi),
\end{equation*}
where $f\colon\mathbb{R}\to (-1,1)$
is a strictly increasing function of class $C^2$ such that the following two conditions hold:
\begin{equation}\label{eq:condf2}
\lim_{\phi\to-\infty}\frac{f''(\phi)}{f'(\phi)}=\lambda_->0, \qquad
\lim_{\phi\to+\infty}\frac{f''(\phi)}{f'(\phi)}=\lambda_+<0.
\end{equation}
In other words, we ``compactify'' the variable $\phi$,
casting it into the set $(-1,1)$.
With this new variable, the system reads as follows:
\begin{subequations}
	\begin{align}
		&\frac{\de s}{\de\tau} = f'(f^{-1}(s)) w,
		\label{eq:mainSys2-a}\\
		&\frac{\de x}{\de\tau}
		=- \frac{x}{\sqrt{6}}
		\left(3w^2 -\kappa z^2\right),
		\label{eq:mainSys2-b}\\
		&\frac{\de w}{\de\tau}
		= \sqrt{\frac{3}{2}}
		w
			\left( 1 -w^2 + \frac\kappa 3 z^2 \right)
			- \frac12\, V'(f^{-1}(s))x^2,
		\label{eq:mainSys2-c}\\
		&\frac{\de z}{\de\tau}
		= \frac{z}{\sqrt{6}}
		\left(1-3w^2 + \kappa z^2\right),
		\label{eq:mainSys2-d}
	\end{align}
\end{subequations}
where now the constraint takes the form:
\begin{equation}
	\label{eq:Constraint2}
	V(f^{-1}(s))x^2 + w^2 - \kappa z^2 = 1.
\end{equation}
The main drawback of this formulation is that it cannot be extended by continuity in $s=\pm 1$,
due to the term $V'(f^{-1}(s))$ in~\eqref{eq:mainSys2-c}.
Indeed, while $f'(f^{-1}(s))$ can be extended by continuity at $s = \pm1$ by setting its value to zero, $V'(f^{-1}(s))$ may possibly diverge.
To overcome this issue, we define the function 
$u\colon (-1,1) \to \mathbb{R}$ as follows:
\begin{equation}
	\label{eq:def-u}
        u(s) = \frac{V'(f^{-1}(s))}{\sqrt{6}\, V(f^{-1}(s))}.
\end{equation}
Exploiting the constraint \eqref{eq:Constraint2},
which is invariant by the flow of \eqref{eq:mainSys2-a}--\eqref{eq:mainSys2-d},
we obtain the following dynamical system:
\begin{subequations}
	\begin{align}
		&\frac{\de s}{\de\tau} = f'(f^{-1}(s))w,
		\label{eq:mainSys3-a}\\
		&\frac{\de x}{\de\tau}
				= -\frac{x}{\sqrt{6}}(3w^2 - \kappa z^2),\label{eq:mainSys3-b}\\
		&\frac{\de w}{\de\tau}
		= \sqrt{\frac{3}{2}}
		\left[w\left(1-w^2 +\frac{\kappa}{3}z^2\right)
			- u(s)(1 - w^2 + \kappa z^2)
		\right],
		\label{eq:mainSys3-c}\\
		&\frac{\de z}{\de\tau}
		= \frac{z}{\sqrt{6}}
		\left(1 - 3w^2 + \kappa z^2 \right).
		\label{eq:mainSys3-d}
	\end{align}
\end{subequations}

To extend the meaning of this system also at infinity and, in general, to obtain the main results of this paper, we introduce the following assumption. 
\begin{assumption}\label{as:V}
    We assume  $V\colon\mathbb{R}\to\mathbb{R}$ is a function of class $\mathcal C^2$ such that:
    \begin{enumerate}
        \item 
         $V(\phi)$ eventually positive  and $V'(\phi)$ eventually negative (resp. positive) 
    for $\phi\to-\infty$ (resp.: $\phi\to+\infty$);  
        \item the two limits
        $\lim_{s\to-1}u(s)$ and $\lim_{s \to 1}u(s)$ exist and are finite.
    \end{enumerate}
\end{assumption}

Observe that the hypotheses made on $V(\phi)$ allow the dynamical system \eqref{eq:mainSys3-a}-\eqref{eq:mainSys3-d}  to be extended by continuity at $s=\pm 1$.
Also observe that $u(s)$ is not well defined on the zeroes of the potential $V$, which however are contained in a compact subset of $(-1,1)$.

\section{The flat case $\kappa = 0$}
\label{sec:the_flat_case_kappa_0_}


In this section, we revisit the case where $\kappa = 0$, which was analyzed in \cite{GGM-JMP2007}, now employing the compactification scheme we previously introduced.
This approach mirrors the methodology applied in \cite{Giambo:2008sa} to the study of homogeneous perfect fluid collapse within $f(R)$ theories.

First, let us rewrite the system~\eqref{eq:mainSys3-a}-\eqref{eq:mainSys3-d} in this particular case:
\begin{subequations}
	\begin{align}
		&\frac{\de s}{\de\tau} = f'(f^{-1}(s))w,
		\label{eq:mainSys4-a}\\
		&\frac{\de x}{\de\tau} = -\sqrt{\frac{3}{2}}\,xw^2,
            \label{eq:mainSys4-b}\\
		&\frac{\de w}{\de\tau}
		= \sqrt{\frac{3}{2}}
		\big(w - u(s)\big)(1 - w^2),
		\label{eq:mainSys4-c}\\
		&\frac{\de z}{\de\tau}
		= \frac{z}{\sqrt{6}}
		\left(1 - 3w^2 \right),
		\label{eq:mainSys4-d}
	\end{align}
\end{subequations}
where now~\eqref{eq:Constraint2} reads as follows:
\begin{equation}
	\label{eq:V-kappa0}
	V(f^{-1}(s))\,x^2 + w^2 = 1.
\end{equation}
Since~\eqref{eq:mainSys4-a}--\eqref{eq:mainSys4-c} and~\eqref{eq:V-kappa0}
don't involve the unknown $z$,
only the first three equations of the system should be considered. 
Moreover,~\eqref{eq:mainSys4-a}
and~\eqref{eq:mainSys4-c} are independent from~\eqref{eq:mainSys4-b}.
However, it is important to notice that the described system is applicable only when $u(s)$ is well defined,
hence when we are sufficiently distant from the zeroes of $V(f^{-1}(s))$.
Should this not be the case,~\eqref{eq:mainSys4-c} must be replaced with the following equation:
\begin{equation}
\label{eq:mainSys4-c2}
\frac{\de w}{\de\tau}
		= \sqrt{\frac{3}{2}}
		w \left( 1 -w^2  \right)
			- \frac12\, V'(f^{-1}(s))\,x^2.
\end{equation}
To overcome this problem, let us fix $\delta > 0$ such that
\begin{equation}
    \label{eq:def-delta}
    V(f^{-1}(s)) > 0 \text{ and } V'(f^{-1}(s)) \ne 0,
\qquad \forall s \in (-1,-1+\delta) \cup (1-\delta,1),
\end{equation}
whose existence is granted by Assumption~\ref{as:V}.
As a consequence,~\eqref{eq:mainSys4-c} is well defined
whenever $|s| \in (1-\delta,1)$.
Let us choose a smooth function $\psi(s)\colon[-1,1]\to[0,1]$ 
such that $\psi(s) = 1$ if $|s| \le 1-\delta$
and $\psi(s) = 0$ if $|s| \in (1-\delta/2,1]$.
In particular, we have that $\psi(s) = 0$
in a right neighbourhood of $s = -1$
and in a left neighbourhood of $s = 1$,
hence where $u(s)$ is well defined, thanks to
Assumption~\ref{as:V}.
Setting $\theta_1(s,w)$ and $\theta_2(s,x,w)$ the right hand-sides of \eqref{eq:mainSys4-c} and~\eqref{eq:mainSys4-c2} respectively, 
we have that both~\eqref{eq:mainSys4-c} and~\eqref{eq:mainSys4-c2}
can be substituted by the following equation:
\begin{equation}\label{eq:mainSys4-c3-k-1}
\frac{\de w}{\de\tau}=(1-\psi(s))\theta_1(s,w)+\psi(s)\theta_2(s,x,w),
\end{equation}
which, with a slight abuse of notation, is always well defined.

\begin{definition}\label{def:csf}
A \textit{collapsing scalar field in the flat case} ($\kappa=0$) is a 
    solution of \eqref{eq:mainSys4-a}, \eqref{eq:mainSys4-b}, and~\eqref{eq:mainSys4-c3-k-1}
    with initial data satisfying \eqref{eq:V-kappa0}, $x(0)>0$, and $V$ satisfying Assumption \ref{as:V}.
\end{definition}

Since we are interested in the asymptotic behaviour of the system,
we give the following result,
recalling that ``generically'' refers to small perturbations of initial data.

\begin{figure}
\centering  
\subfigure[\ ]{\includegraphics[width=0.45\linewidth]{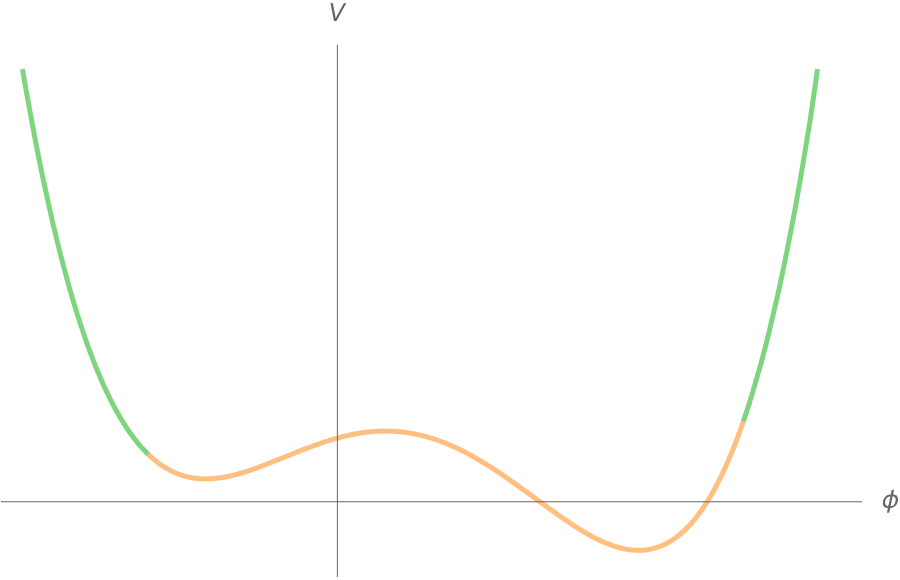}}
\subfigure[\ ]{\includegraphics[width=0.45\linewidth]{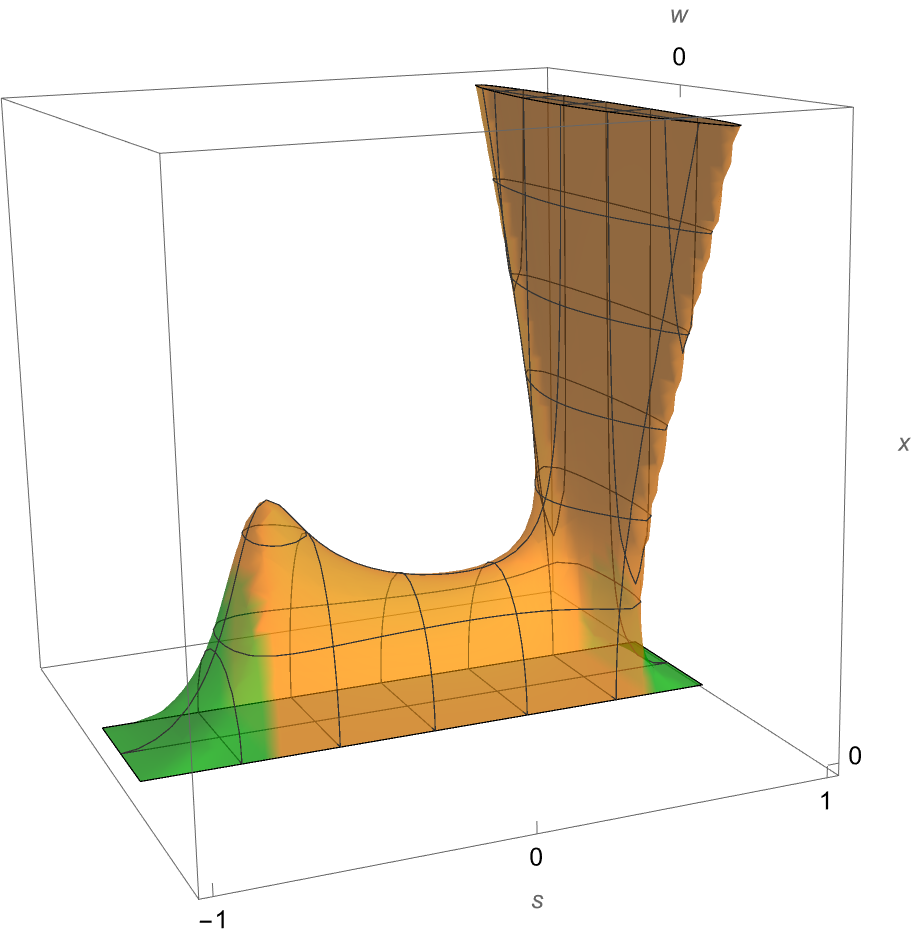}}
\subfigure[\ ]{\includegraphics[width=0.45\linewidth]{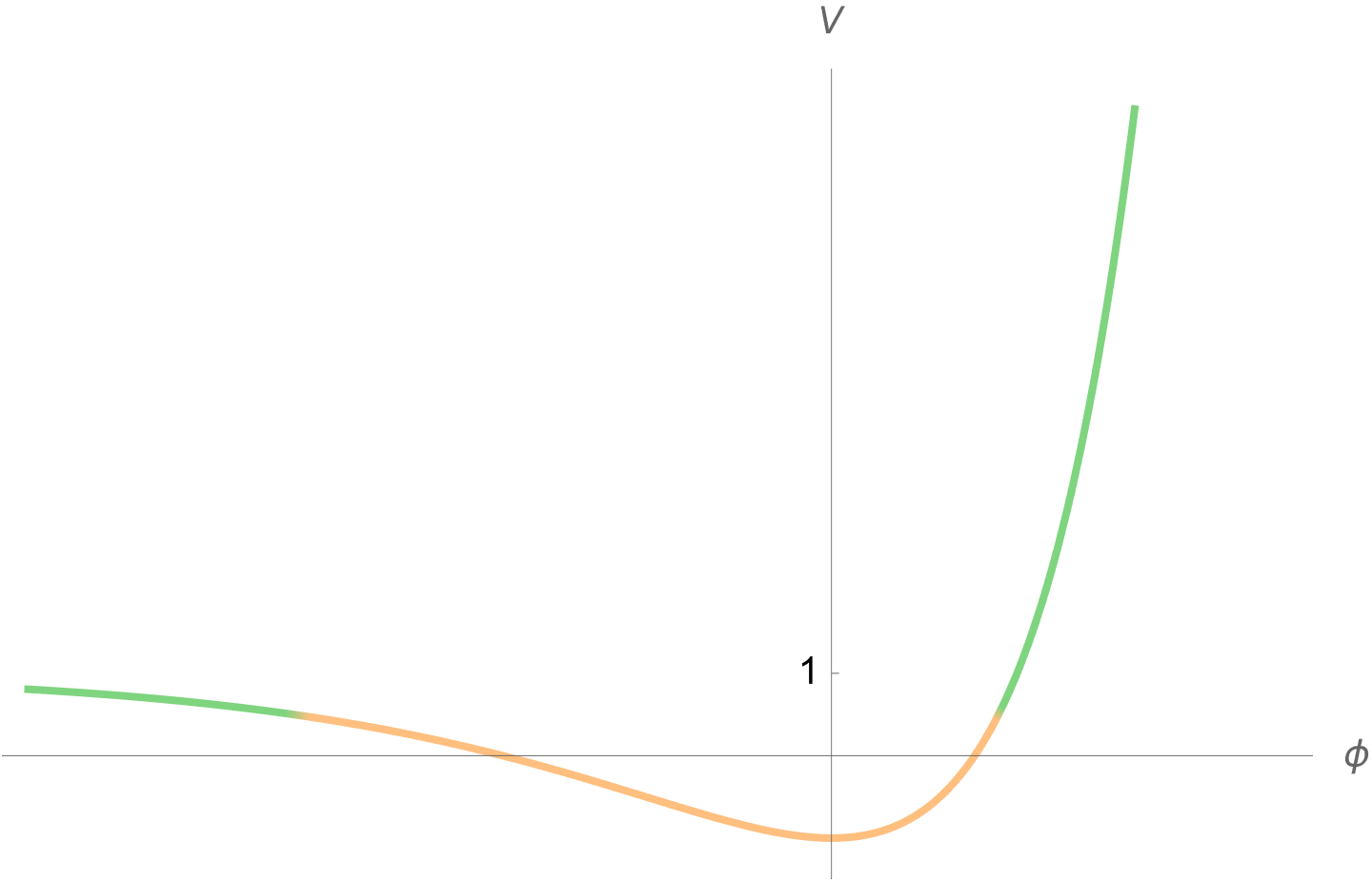}}
\subfigure[\ ]{\includegraphics[width=0.45\linewidth]{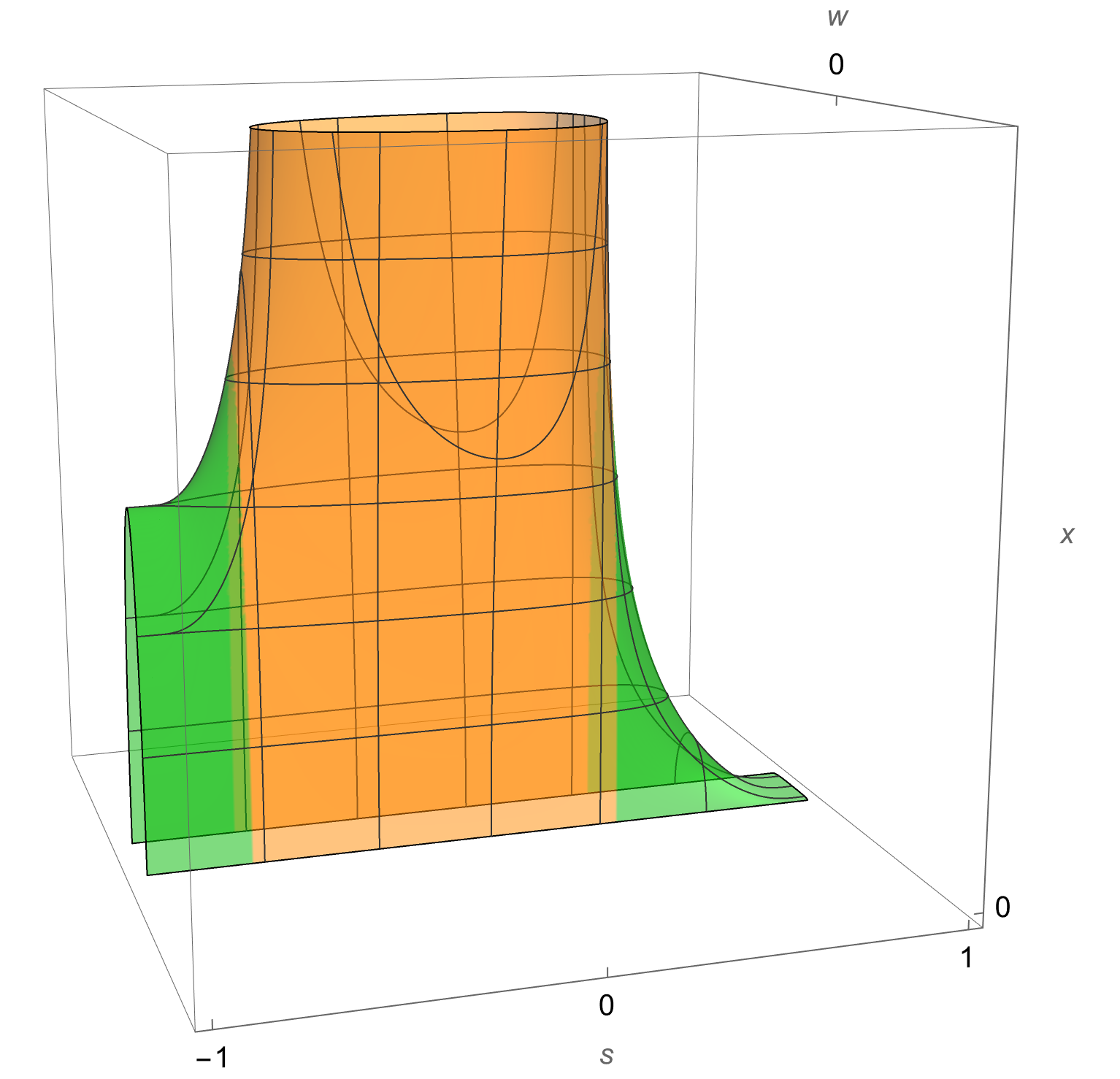}}
\subfigure[\ ]{\includegraphics[width=0.45\linewidth]{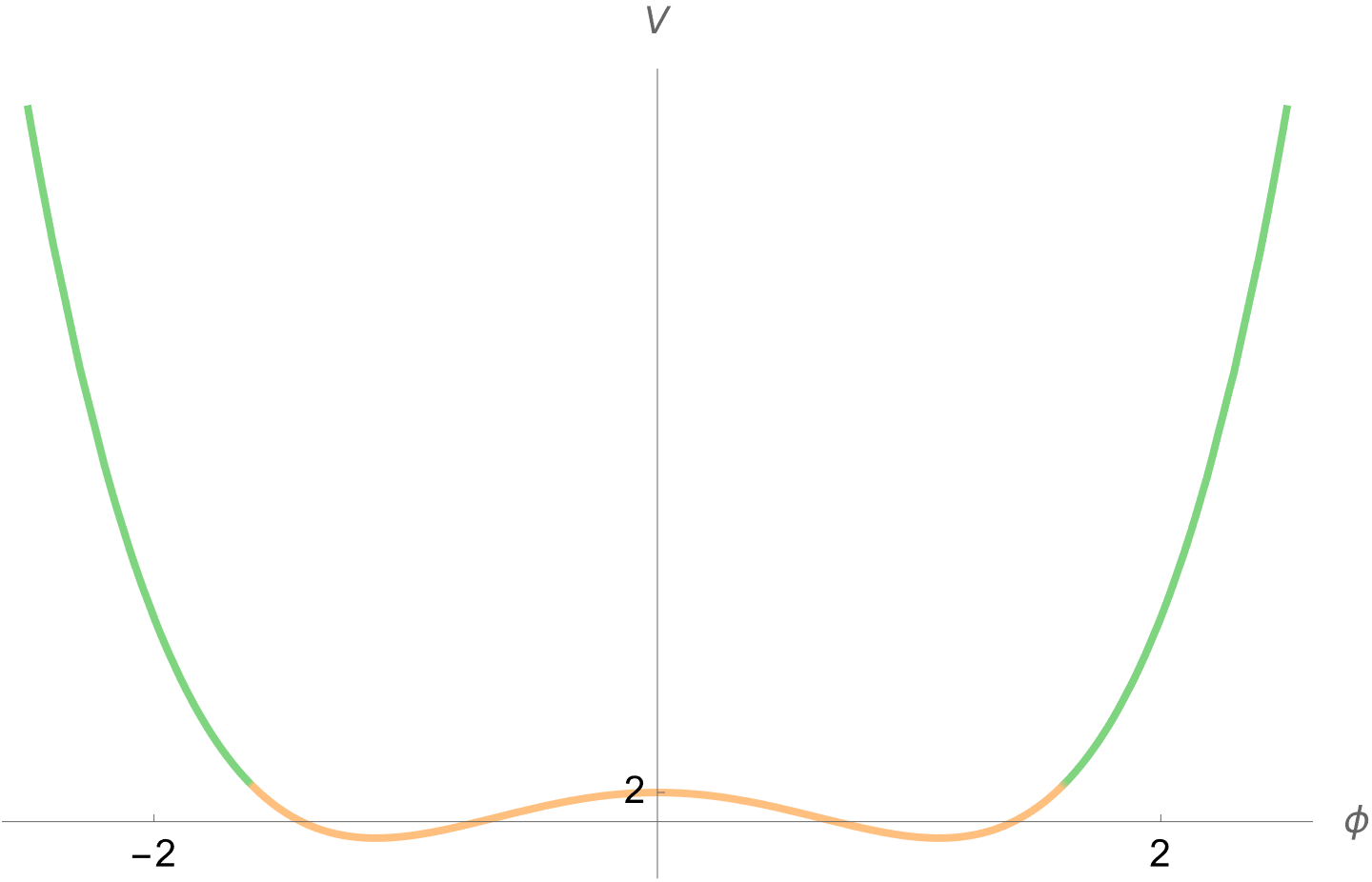}}
\subfigure[\ ]{\includegraphics[width=0.45\linewidth]{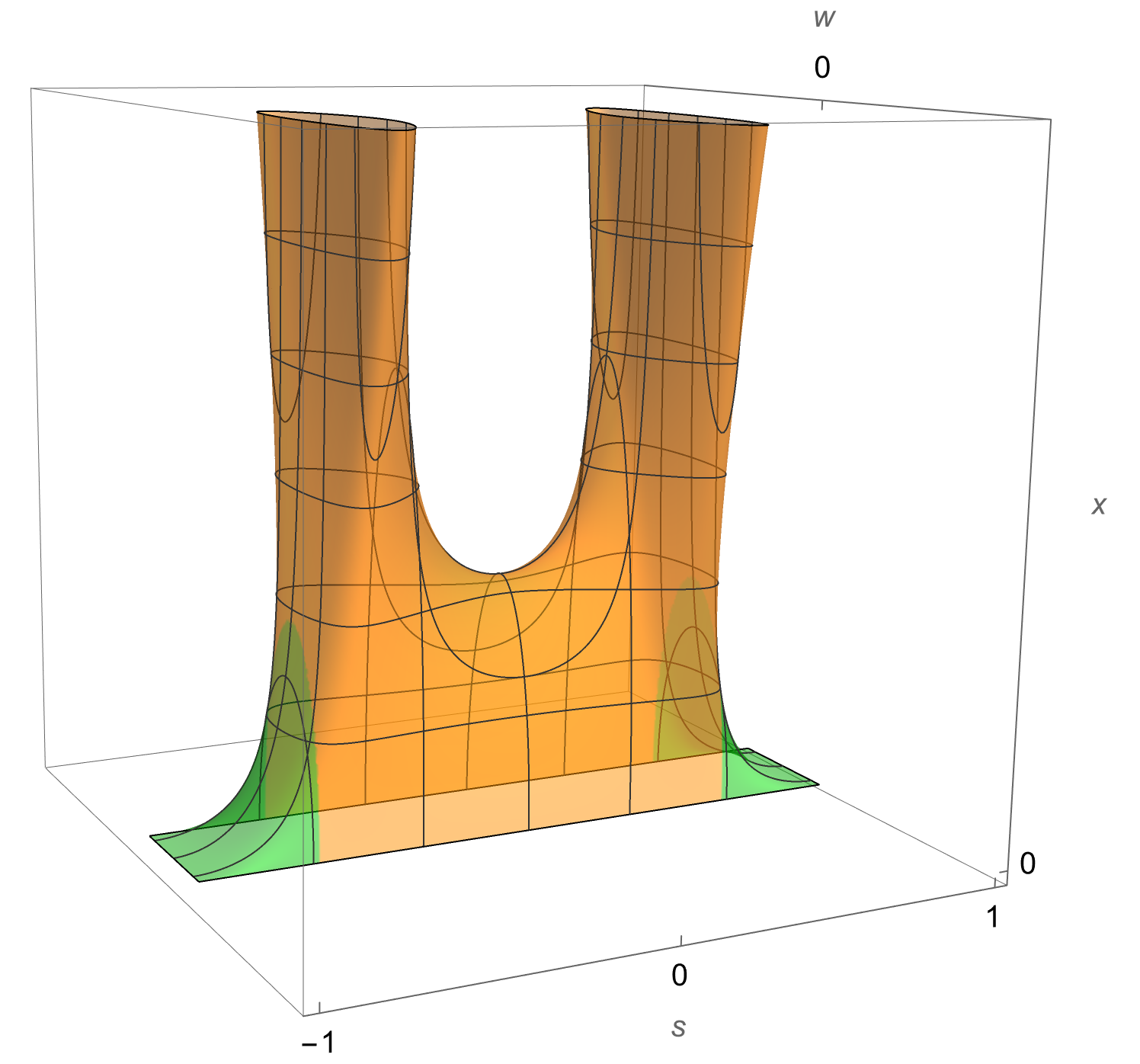}}
\caption{The subset $\Omega$ \eqref{eq:def-Omega}, for several choices of the potential $V(\phi)$. The orange region corresponds to the set where $\psi(s)=1$, see \eqref{eq:def-delta} and the ensuing discussion.}
\label{fig:M}
\end{figure}

\begin{teo}
	\label{theorem:x0}
	Let $\bar{x} > 0$ be fixed and let $\Omega$
	(see Figure \ref{fig:M}) be defined as the     following subset of $\mathbb{R}^3$:
	\begin{equation}
		\label{eq:def-Omega}
		\Omega = \Big\{
			(s,x,w) \in \mathbb{R}^3: s \in (-1,1),\, x \in (0,\bar{x}],\,
            V(f^{-1}(s))x^2 + w^2 = 1
		\Big\}.
	\end{equation}
Assume that $V$ satisfies Assumption \ref{as:V}.
 Then, a collapsing scalar field in the flat case,
	with initial data in $\Omega$, remains in $\Omega$
	for all $\tau \ge 0$,
	and generically approaches the set $M_0\subseteq \overline{\Omega}\setminus\Omega$ defined by:
	\begin{equation}
		\label{eq:boundaryOmega}
	   M_0\coloneqq\overline\Omega\cap \Big(\{|s|=1\}\cup\{|w|=1\}\Big)\cap\{x=0\}.
	\end{equation}
	In particular, we have generically that
	\begin{equation}
		\label{eq:lim-xtau}
		\lim_{\tau\to\infty}x(\tau) = 0^+.
	\end{equation}
\end{teo}

\begin{proof}
    The results are consequences of applying the LaSalle theorem (see, e.g.,~\cite{khalil}).
    We divide the proof into several steps,
    verifying the hypoteses of that theorem.

\textbf{Step 1.} 
Let us prove that $\Omega$ is a bounded set.
By Assumption~\ref{as:V}, $V$ is bounded from below
and, by~\eqref{eq:mainSys4-b}, $x(\tau)$ is a decreasing function,
hence $x(\tau)\in [0,\bar x]$ for all $\tau\ge 0$ if $x(0)\in[0,\bar x]$.
Then, the constraint \eqref{eq:V-kappa0} implies that $w(\tau)$ is also bounded for all $\tau\ge 0$.

\textbf{Step 2.}
Let us now determine the closure $\overline\Omega$, which is a compact set.
Limit points $(s,x,w)$ for $\Omega$ that are outside $\Omega$ may be of the following forms:
\begin{itemize}
    \item $s=\pm 1$.
    If 
    \[
    V^+_{\infty}:=\lim_{\phi\to+\infty}V(\phi)=+\infty,
    \]
    then $(1,0,w)$ for any $w\in[-1,1]$ is a limit point for sequences in $\Omega$. On the contrary, if $V^+_\infty\in\mathbb R$ then all the points $(1,x,w)$ such that $V^+_\infty x^2+w^2 = 1$
    and $x\ge 0$ are limit points.
    Analogous considerations hold for $V^-_\infty:=\lim_{\phi\to-\infty}V(\phi)$.
    \item $x=0$. Besides the situations already considered before, there are also the points $(s,0,\pm 1)$ with $s\in(-1,1)$. 
\end{itemize}
Summarizing, the closure $\overline\Omega$ is obtained by adding the set
\[
\big\{(\pm 1,x,w)\in \mathbb{R}^3\,:\,x\in[0,\bar x],\,V^\pm_\infty x^2+w^2=1\big\}\cup \big\{(s,0,\pm 1) \in \mathbb{R}^3\,:s\in(-1,1)\big\}
\]
to $\Omega$, where, with abuse of notation, we are considering the case $w\in[-1,1]$, $x=0$  as a degenerate case ($V^\pm_\infty=\infty$) of the the half-ellipse $V^\pm_\infty x^2+w^2=1$, $x\ge 0$.

\textbf{Step 3.} Now we prove that $\overline\Omega$ is positively invariant. This can be done simply by considering when the initial condition is either in $\Omega$ or not. 
In the first case, we already know that the constraint \eqref{eq:V-kappa0} is satisfied and $x(\tau)$ is decreasing, so the solution curve will remain in $\Omega$.
Hence, it suffices to see what happens when the curve starts from $\overline{\Omega}\setminus\Omega$. 

If $s(0)=+1$, since $f'(f^{-1}(1))=0$ then $s(\tau)=1$ for all $\tau\ge 0$.
If $V^+_\infty\in\mathbb R$ then $u(1)=0$ and it can be seen that the following identity holds.
\[
\frac{\mathrm d}{\mathrm d\tau}(V^+_\infty x^2+w^2)=
-\sqrt 6 w^2 (V^+_\infty x^2+w^2 -1) = 0.
\] 
Hence, the constraint is conserved along the flow
and the solution remains in $\overline\Omega$.
Otherwise, if $V^+_\infty=\infty$
then $x(0)=0$, so $x(\tau)= 0$ for every $\tau\ge 0$.
Therefore,~\eqref{eq:mainSys4-c} implies that $w(\tau)\in[-1,1]$ for all $\tau\ge 0$.
Analogous arguments hold for $s=-1$. 

If the solution starts from a point $(s,0,\pm 1)$ with $s\in(-1,1)$,
then $x(\tau)=0$ and $w(\tau)=\pm 1$ for all $\tau\ge 0$, while~\eqref{eq:mainSys4-a} implies that $s(\tau)$ will be always in $(-1,1)$.
This completes the proof of the positive invariance of $\overline\Omega$. 

\textbf{Step 4.}
In view of the application of LaSalle theorem,
let us notice that the projection on $x$ is a Lyapunov function on  $\overline{\Omega}$, 
so we need to determine the subset $E$ of $\overline\Omega$
where $\de{x}/\de\tau = 0$.
Recalling \eqref{eq:mainSys4-b}, $E$ is the set of points where $w$ or $x$ vanish, and it consists of the following subsets:
\begin{itemize}
    \item\label{itm:1} $
    E_1 = \big\{(s,1/\sqrt{V(f^{-1}(s))},0): s \in (-1,1)\big\}$,
    and the limit case at infinity when $V^{\pm}_\infty \in \mathbb{R}$,
    hence
    $E_2 = \big\{(\pm1, 1/\sqrt{V^{\pm}_\infty},0)\big\}$;
    \item\label{itm:2} $E_3 = \big\{(\pm 1, 0, w): w \in [-1,1]\big\}$,
    when $V^\pm_\infty=\infty$;
    \item\label{itm:3} $E_4 = \big\{(s,0,\pm 1): s\in[-1,1]\big\}$.
\end{itemize}

\textbf{Step 5.}
We need to characterize the largest invariant subset of $E$.
We observe that $E_3$ and $E_4$ above both belong to this set: indeed, if $x(0)=0$ then $x(\tau)=0$ for all $\tau\ge 0$.

Now, let us examine the subset $E_1$.
When $s_0\in(-1,1)$, then the solution with this initial condition will remain in $E_1$ if  $w(\tau)=0$ for all $\tau\ge 0$.
This implies $s(\tau)=s_0$ and $x(\tau)=1/\sqrt{V(f^{-1}(s_0)}$.
Since $w(\tau)=0$ for every $\tau \ge 0$,
by~\eqref{eq:mainSys4-c2}, or by~\eqref{eq:mainSys4-c}, we deduce that $V'(f^{-1}(s_0))=0$, hence $s_0$ corresponds to a critical point of the potential $V$ with positive critical value. 

By considering the subset $E_2$,
if $s_0=\pm 1$  and $V^\pm_\infty\in\mathbb R$ we have a critical point ``at infinity'' for $V(\phi)$.
Since $f'(f^{-1}(\pm 1))=0$ by assumption, then $s(\tau)=s_0$ for every $\tau \ge 0$.
Moreover, since $u$ is always well defined on a neighbourhood of the extreme points $s = \pm 1$, and in this case $u(\pm1)= V'(\pm\infty)/(\sqrt{6}V^{\pm}_\infty) = 0$, 
we can use~\eqref{eq:mainSys4-c} to deduce that $w(\tau)=0$ for all $\tau\ge 0$. 

Summarizing, the largest invariant subset of $E$ is made by the subsets $E_3$ and $E_4$,
which together constitute the set $M_0$ defined by~\eqref{eq:boundaryOmega},
plus the subset of $E_1 \cup E_2$
for which $f^{-1}(s)$ is a critical point for $V$ with positive critical value, possibly including the case of critical points ``at infinity'' $s=\pm 1$,
hence
\[
    E_0 = \big\{
            (s,1/\sqrt{V(f^{-1}(s))},0):
            s \in [-1,1],\, V'(f^{-1}(s)) = 0
            \big\}.
\]
\textbf{Step 6.} 
Application of LaSalle invariance theorem says that a solution of the system approaches at infinity the set determined in the step above. 
To complete the proof, we have to show that the points in $E_0$,
which are all critical points for the system, are unstable.
Let us consider the case when $s = \pm1$ and $V^{\pm}_\infty \in \mathbb{R}$.
In this case, since we are in a neighbourhood of an extreme point, we can decouple the system, considering just~\eqref{eq:mainSys4-a} and~\eqref{eq:mainSys4-c}.
The linearized system at the point $(\pm1,0)$ is characterized by the following matrix:
\[
    \begin{bmatrix}
       0 & 0 \\
       -\sqrt{\frac{3}{2}}u'(\pm1) & \sqrt{\frac{3}{2}}
    \end{bmatrix},
\]
which always admits a positive eigenvalue.
Therefore, these points are unstable.
When a point in $E_0$ is such that $s \ne \pm1$,
then we have to work with the system \eqref{eq:mainSys4-a},\eqref{eq:mainSys4-b}, and ~\eqref{eq:mainSys4-c2}
in $(s,x,w)$. 
Nevertheless, a similar linearization procedure
shows the instability of such a point,
since there always exists an eigenvalue with positive real part.

Then, solutions must generically approach the set $M_0$ defined by~\eqref{eq:boundaryOmega},
and~\eqref{eq:lim-xtau} holds.
\end{proof}

According to the above theorem, generically, the study of limit points of trajectories within the set $M_0$ identifies the candidates for $\omega$-limit points of the system that we are considering.
First, let us notice that all the points in $M_0$
such that $s\ne \pm1$ are not equilibrium points.
Indeed, these points should have $|w| = 1$,
hence~\eqref{eq:mainSys4-a} implies that $s(\tau)$
is not constant.
Hence, we can reduce our analysis on the neighbourhoods of the extreme points where $|s| = 1$ and we can work with~\eqref{eq:mainSys4-c}.
By reducing once again to the planar system given by ~\eqref{eq:mainSys4-a} and~\eqref{eq:mainSys4-c},
the linearization of the system at a point $(s,w)$ gives the following matrix:
\begin{equation}
\begin{bmatrix}
 w\frac{f''\left(f^{(-1)}(s)\right)}{f'\left(f^{(-1)}(s)\right)}
 & f'\left(f^{(-1)}(s)\right) \\
 \sqrt{\frac{3}{2}} \left(w^2-1\right) u'(s)
 & \sqrt{\frac{3}{2}} \left(2 u(s) w-3 w^2+1\right)
\end{bmatrix}.
\end{equation}
Let us explore the possible cases corresponding to the equilibrium points of ~\eqref{eq:mainSys4-a},~\eqref{eq:mainSys4-c}:
\begin{enumerate}
	\item $s=\pm1$, with $u(1) = 0$ or $u(-1) = 0$ and $w=0$. The eigenvalues are 0 and $\sqrt{3/2}$, hence these points are unstable;
	\item $s  = \pm 1$ and $w = \pm 1$. We can limit to the two cases where $w$ and $s$ have the same sign, since they are the physically relevant ones ($\phi$ and $\dot\phi$ have the same sign). 
    Considering for instance $s=w=1$, the hypothesis \eqref{eq:condf2} made on $f$ gives the two eigenvalues $\lambda_+ <0$ and $\sqrt 6 (u(1)-1) $. Therefore, $u(1)<1$ is a sufficient condition for stability. Same argument applies to $s=w=-1$, where $u(-1)>-1$ is needed to have a stable equilibrium point.
 \item $s = \pm 1$ and $w = u(\pm1)$. Observe that, due to Assumption \ref{as:V} made on $V(\phi)$, the constraint \eqref{eq:V-kappa0} implies $u(\pm 1)^2\le 1$. Considering the positive case $s=1$, the two eigenvalues are given by $\lambda_+ u(1) <0$ and $-\sqrt{3/2} \left(u(1)^2-1\right)$, then if $u(1)<1$ the equilibrium is unstable.
 Similarly, if $u(-1) > -1$, then the point
 $(-1,u(-1))$ is unstable.
\end{enumerate}

We summarize the above considerations in the following result.
\begin{proposition}
    \label{prop:sf}
    If $u(-1)>-1$ or $u(1)<1$, then the scalar field generically diverges to infinity, i.e., $s^2\to 1$, with a velocity such that $w^2\to 1$.
    
\end{proposition}

\begin{example}\label{ex:tanh}
We can observe that when both $u(-1)<-1$ and $u(1)>1$, the above system lacks stable equilibrium points.
In this case, it can be shown that the solution approaches the set $M$ as a limit cycle. As an example, let us consider the following potential:
$$
V(\phi)=\cosh \left( \alpha\sqrt{6}  \phi \right),
$$
where $\alpha$ is a positive parameter. In this situation we choose the function
$$
s=f(\phi)=\tanh \left( \alpha  \sqrt{6}\phi \right),
$$
in such a way that $u(s)$ becomes simply $u(s)=\alpha s$ and the system given by \eqref{eq:mainSys4-a} and \eqref{eq:mainSys4-c} takes the following the form:
\begin{equation}\label{eq:sys-cosh}
\frac{\de s}{\de\tau}=\alpha \sqrt{6}  \left(1-s^2\right) w,\qquad\frac{\de w}{\de\tau}=\sqrt{\frac{3}{2}} \left(1-w^2\right) (w-\alpha  s).
\end{equation}

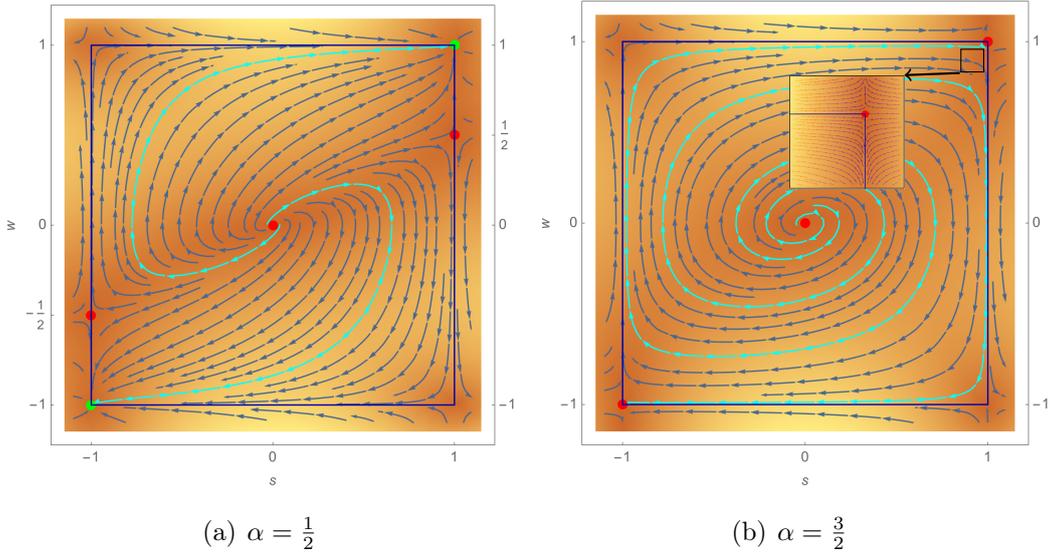
\begin{figure}
    \centering
\subfigure[$\alpha=\frac12$]{\includegraphics[width=0.45\linewidth]{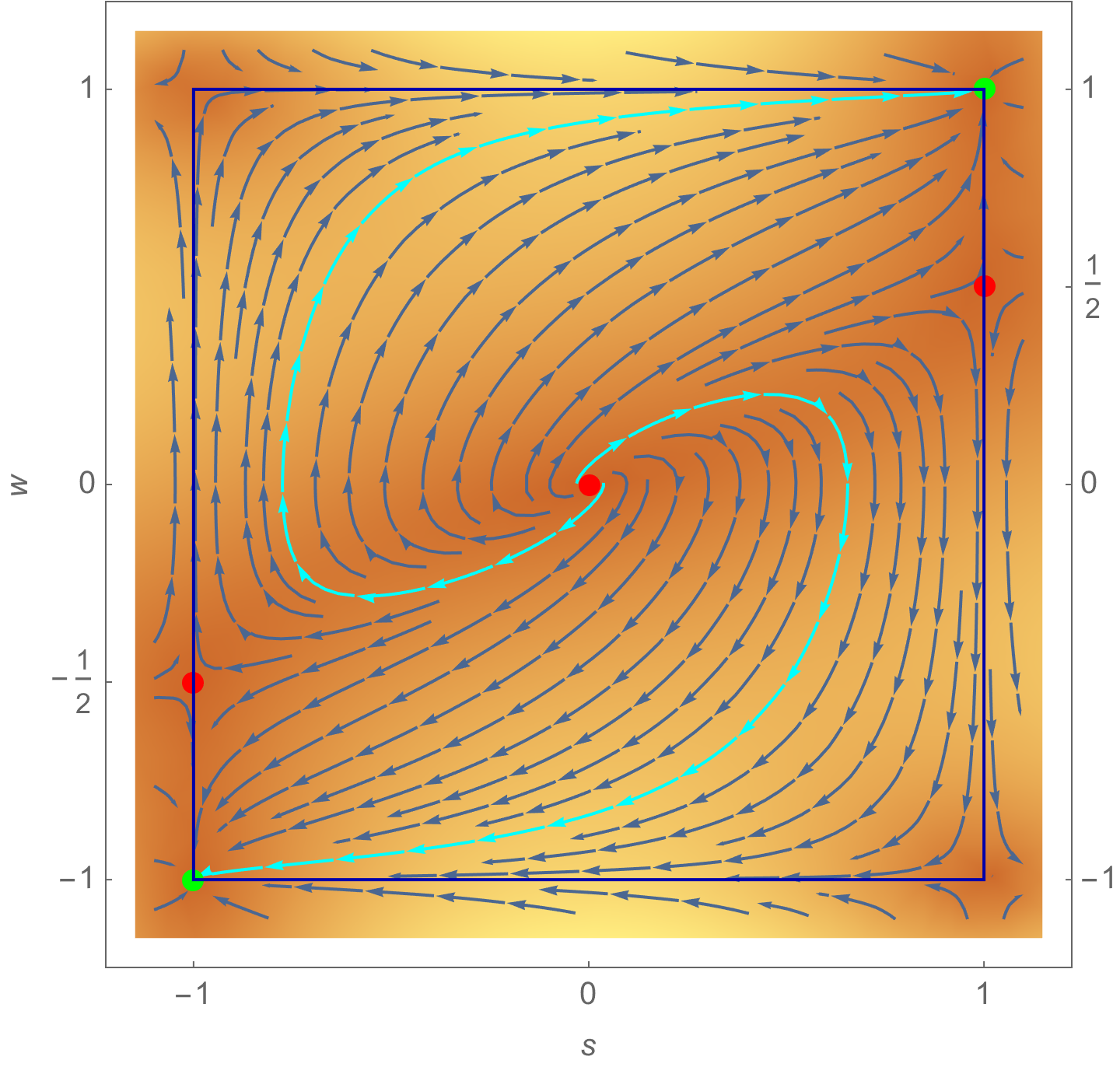}}
\subfigure[$\alpha=\frac32$ ]{
\begin{tikzpicture}
\node[inner sep=0pt] (main) at (0,0)
    {\includegraphics[width=0.45\textwidth]{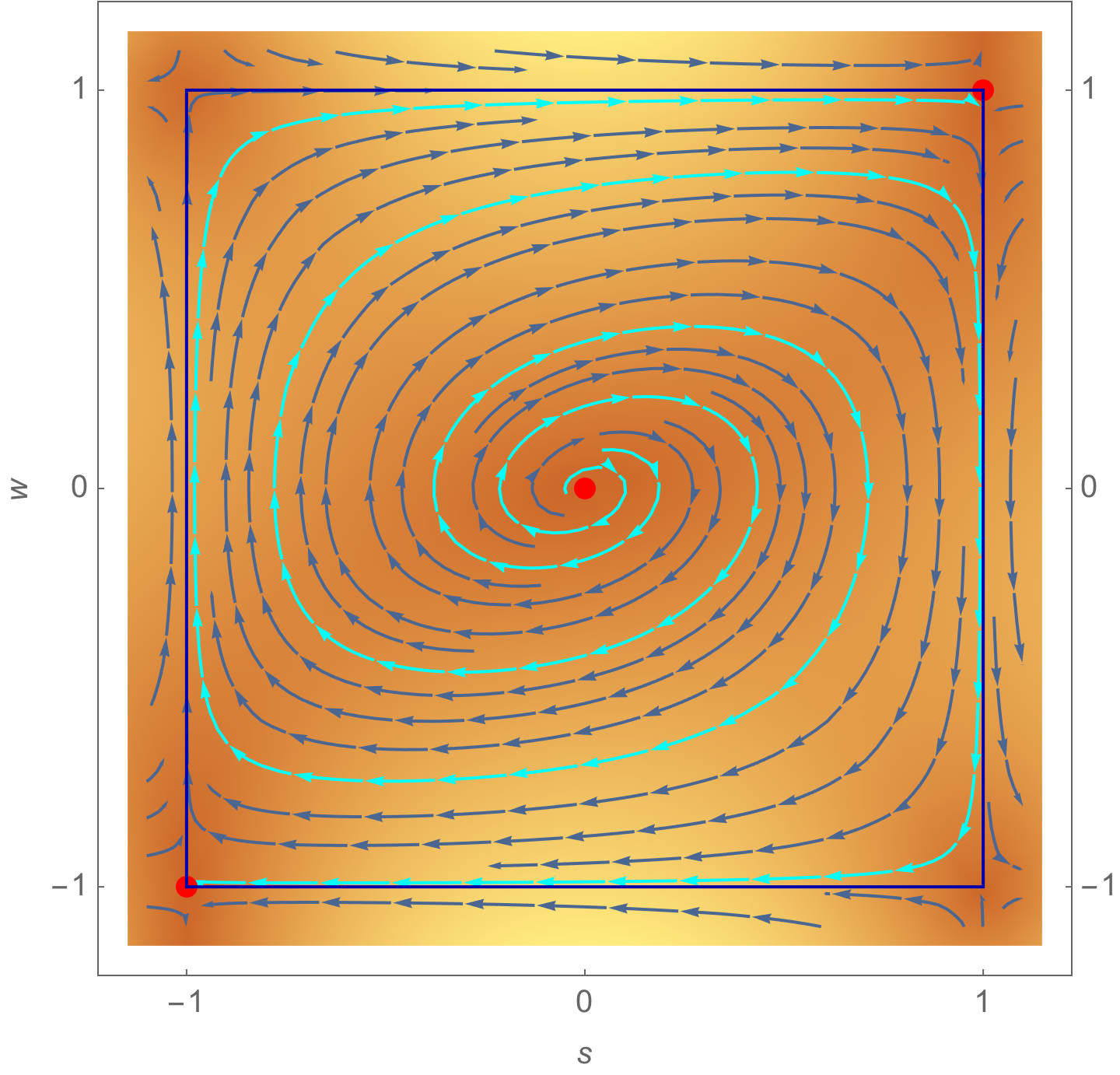}};
\node[inner sep=0pt] (son) at (0.7,1.5)
    {\includegraphics[width=.1\textwidth]{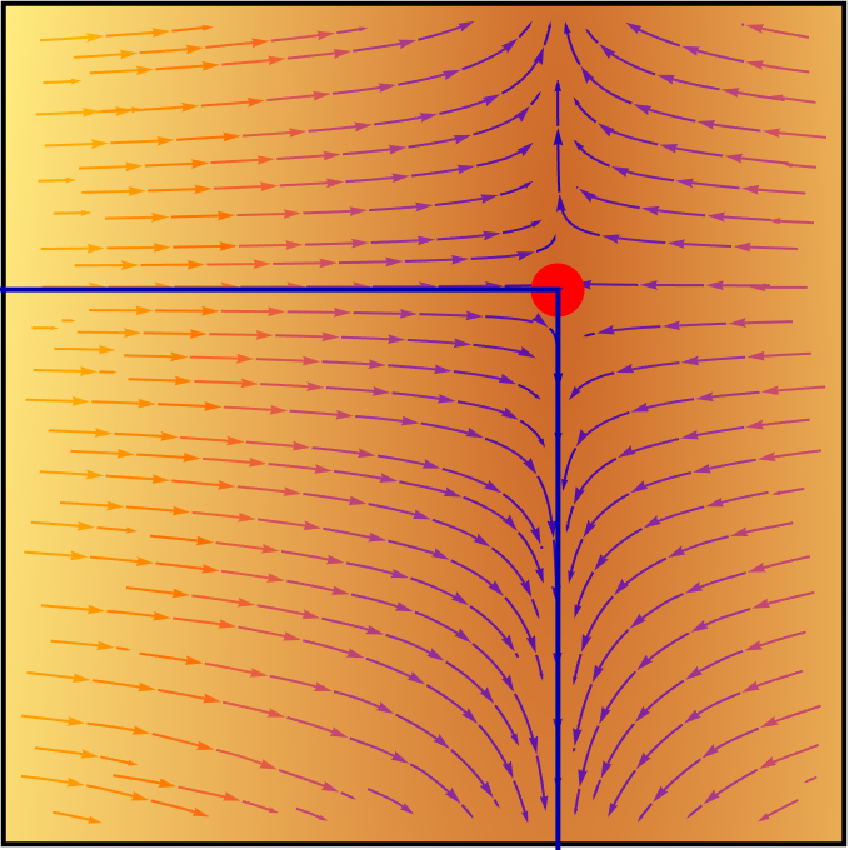}};
\node (rect) at (2.33,2.43) {\ };
\draw[draw=black] (2.2,2.3) rectangle ++(0.3,0.3);
\draw[<-,thick] (son.north east) -- (rect.south west);
\end{tikzpicture}
  } 
\caption{The flow of the vector field \eqref{eq:sys-cosh} in Example \ref{ex:tanh} for two values of $\alpha$. When $\alpha\ge 1$ the point $(1,1)$ (and $(-1,-1)$ as well) become unstable equilibria.}
    \label{fig:ex}
\end{figure}

Therefore, when $\alpha<1$ we have $(s,w)=(1,1)$ and $(s,w)=(-1,-1)$ as stable equilibria, but when $\alpha>1$ the only equilibria are the origin and the vertices of the square $[-1,1]\times[-1,1]$, and none of them is stable (see Figure \ref{fig:ex}).

The case when $\alpha = 1$ proves more delicate to analyze, given that $(-1,-1)$ and $(1,1)$ represent semi-hyperbolic equilibria. In this scenario, it is possible to demonstrate that integral curves with initial conditions inside the square $[-1,1]\times[-1,1]$ never converge to these equilibria (see, e.g., \cite[Theorem 2.19(iii)]{dumortier2006} for more details).
Nevertheless, in all the above cases the integral curves must approach the boundary of the square, as one can infer also by calculation of the function $x^2(\tau)$, that we know to be decreasing.
In this case, it reads as follows:
\[
x^2(\tau)=\sqrt{1-s^2} \left(1-w^2\right).
\]
Since the integral curves of \eqref{eq:sys-cosh} with initial condition inside the square approaches its boundary,
the scalar field oscillates infinitely many times with wider and wider amplitudes around $\phi=0$. 
\end{example}

\section{The open spatial topology $\kappa = -1$}%
\label{sec:collapse_of_the_open_spatial_topology_kappa_1_}
For reader's convenience, we report here the system 
in its complete form, with $\kappa = -1$:
\begin{subequations}
	\begin{align}
		&\frac{\de s}{\de\tau} = f'(f^{-1}(s))w,\label{eq:mainSys-5a}\\
		&\frac{\de x}{\de\tau}
		= -\frac{x}{\sqrt{6}}(3w^2+ z^2),\label{eq:mainSys-5b}\\
		&\frac{\de w}{\de\tau}
		= \sqrt{\frac{3}{2}}
		\left[w\left(1 - w^2 - \frac{1}{3}z^2\right)
		-u(s)(1-w^2 - z^2)\right],\label{eq:mainSys-5c}\\
		&\frac{\de z}{\de\tau}
		= \frac{z}{\sqrt{6}}
		\left(1 - 3w^2 -z^2\right),\label{eq:mainSys-5d}
	\end{align}
\end{subequations}
where the constraint 
(invariant w.r.t. the flow of the system)
reads as follows:
\begin{equation}\label{eq:constr-wneg}
	V(f^{-1}(s))x^2 + w^2 + z^2 = 1.
\end{equation}
As for the case $\kappa = 0$,
near the zeros of the potential~\eqref{eq:mainSys-5c} should be replaced by
\begin{equation}\label{eq:mainSys5-cbis}
    \frac{\de w}{\de\tau}
    = \sqrt{\frac{3}{2}}w
	\left( 1 -w^2 -\frac13 z^2 \right)
	- \frac12\, V'(f^{-1}(s))x^2.
\end{equation}  
Setting $\theta_1(s,w,z)$ and $\theta_2(s,x,w,z)$ as the right hand sides of \eqref{eq:mainSys-5c} and \eqref{eq:mainSys5-cbis}, respectively, we can repeat the construction leading to Definition \ref{def:csf}.
For this case, we give the following definition.
\begin{definition}\label{def:csf-k-1}
    A \textit{collapsing scalar field in the  case} $\kappa=-1$ is a 
    solution of \eqref{eq:mainSys-5a}, \eqref{eq:mainSys-5b},  
\begin{equation}\label{eq:mainSys4-c3}
\frac{\de w}{\de\tau}=(1-\psi(s))\theta_1(s,w,z)+\psi(s)\theta_2(s,x,w,z),
\end{equation}
and \eqref{eq:mainSys-5d} with initial data satisfying \eqref{eq:constr-wneg}, $x(0)>0$,
and $V$ satisfying Assumption~\ref{as:V}.
\end{definition}

We can now provide the counterpart of Theorem~\ref{theorem:x0}.
\begin{teo}
	\label{theorem:x0-wneg}
	Let $\bar{x} > 0$ be fixed and let $\Omega$
	be the following subset of $\mathbb{R}^4$:
	\begin{equation}
		\label{eq:def-Omega-wneg}
		\Omega = \Big\{
			(s,x,w,z) \in \mathbb{R}^4: s \in (-1,1),\, x \in (0,\bar{x}],\,
            V(f^{-1}(s))x^2 + w^2+z^2 = 1
		\Big\}.
	\end{equation}
If Assumption~\ref{as:V} holds,
 then a collapsing scalar field in the case $\kappa=-1$ with initial data in $\Omega$ remains in $\Omega$
	for every $\tau \ge 0$,
	and generically approach the set $M_{-1}\subseteq \overline{\Omega}\setminus\Omega$ defined as follows:
	\begin{equation}
		\label{eq:boundaryOmega-wneg}
	   M_{-1}\coloneqq\overline\Omega\cap \Big(\{|s|=1\}\cup\{w^2+z^2=1\}\Big)\cap\{x=0\}.
	\end{equation}
	In particular, we have generically that
	\[
		\lim_{\tau\to\infty}x(\tau) = 0^+.
	\]
\end{teo}

\begin{proof}
    As for Theorem~\ref{theorem:x0},
    the proof is based on the LaSalle invariance theorem and it proceeds following the same steps.

    \textbf{Step 1.} Since Assumption~\ref{as:V} implies that $V$ is bounded from below, it is easy to see that $\Omega$ is bounded.
    
    \textbf{Step 2.} In this case, the limit points for $\Omega$ that are not in $\Omega$
    are the following:
    \begin{multline*}
\overline\Omega\setminus\Omega=\{(\pm 1,x,w,z)\,:\,x\in[0,\bar x],\,V^\pm_\infty x^2+w^2+z^2=1\} \\ 
\cup \{(s,0,w,z)\,:\,s\in(-1,1),\,w^2+z^2=1\},
\end{multline*}
where again, with abuse of notation, we include in the first subset the situation where $V^\pm_\infty=\infty$, which implies
$ x=0$ and $w^2+z^2\le 1$.

\textbf{Step 3.}
Let us show the positive invariance of $\overline{\Omega}$.
If the initial condition is inside $\Omega$,
then~\eqref{eq:constr-wneg} is satisfied and $x$
is a decreasing function, so the solution remains inside $\Omega$.
Otherwise, we need to consider the following situations:
\begin{itemize}
\item $s_0 = \pm1$ and $V^{\pm}_\infty \in \mathbb{R}$:
in this case we have $u(s_0) = 0$. Hence, by using~\eqref{eq:mainSys-5a}--\eqref{eq:mainSys-5d}
one can find that 
\begin{equation}\label{eq:const}
\frac{\de}{\de\tau}(V^\pm_\infty x^2+w^2+z^2)=-\sqrt{\frac{2}{3}}(3w^2+z^2)(V^\pm_\infty x^2+w^2+z^2-1)
= 0,
\end{equation}
so the constraint is satisfied for all $\tau\ge 0$.
\item $s_0 = \pm1$ and $V^{\pm}_\infty = \infty$:
in this case we have $x_0 = 0$ and $w_0^2 + z_0^2 \le 1$.
Since $f'(f^{-1}(s_0)) = 0$ and $x(\tau) = 0$
for all $\tau \ge 0$, the positive invariance is achieved by showing that $w^2(\tau) + z^2(\tau) \le 1$
for all $\tau \ge 0$.
Since we are on a neighbourhood of the extrema, we can use~\eqref{eq:mainSys-5c} that, together with~\eqref{eq:mainSys-5d}, implies the following equation:
\begin{equation}
    \label{eq:w^2+z^2le1}
    \frac{\de}{\de\tau}(w^2 + z^2) 
= \sqrt{\frac{2}{3}}(1 - w^2 - z^2)(3w^2 + z^2 - 3u(s_0)w).
\end{equation}
By the last identitty, we deduce that if $w_0^2 + z_0^2 \le 1$ then $w^2(\tau) + z^2(\tau) \le 1$
for every $\tau \ge 0$.
\item $x_0 = 0$, $s_0 \in (-1,1)$ 
and $w_0^2 + z_0^2 = 1$:
in this case, it suffices to show that 
$w^2(\tau) + z^2(\tau) = 1$ for every $\tau \ge 0$.
By combining~\eqref{eq:mainSys-5d} 
and~\eqref{eq:mainSys5-cbis} we obtain
\begin{equation}
 \label{eq:w^2+z^2eq1}
    \frac{\de}{\de\tau}(w^2 + z^2) 
= \sqrt{\frac{2}{3}}(1 - w^2 - z^2)(3w^2 + z^2),
\end{equation}
hence the constraint is always satisfied.
\end{itemize}

\textbf{Step 4.}
Even in this case the function $x(\tau)$ is a Lyapunov function and we need to characterize the set $E\subset\overline{\Omega}$
where $\de x/\de\tau = 0$.
By~\eqref{eq:mainSys-5b} we obtain that $E$
is constituted by the following subsets:
\begin{itemize}
    \item $E_1 = \big\{(s,1/\sqrt{V(f^{-1}(s)},0,0): s \in [-1,1]\big\}$, where the limit cases $s = \pm 1$
    are possible only when $V^{\pm}_\infty \in \mathbb{R}$;
    \item $E_2 = \big\{(\pm1,0,w,z): w^2 + z^2 \le 1\big\}$ if $V^{\pm}_{\infty} = \infty$;
    \item $E_3 = \big\{(s,0,w,z): s \in (-1,1),\, w^2 + z^2 = 1\big\}$;
\end{itemize}

\textbf{Step 5.}
Let us determine the largest invariant subset of $E$.
The two subsets $E_2$ and $E_3$ are invariant by~\eqref{eq:w^2+z^2le1} and~\eqref{eq:w^2+z^2eq1}, respectively.
A solution starting in $E_1$ remains in $E$ only if it remains in $E_1$.
Hence, $w(\tau)$ should be zero for every $\tau$
and by~\eqref{eq:mainSys5-cbis} we obtain that $V'(f^{-1}(s_0))$ should vanish.
As a consequence, the largest invariant subset of $E$
is constituted by $E_2$ and $E_3$, which together form the set $M_{-1}$, and the subset $E_0 = \big\{(s,1/\sqrt{V(f^{-1}(s))},0,0): s \in [-1,1],\, V'(f^{-1}(s)) = 0 \big\}$.
            
\textbf{Step 6.}
The proof ends by noticing that every point in $E_0$
is an unstable fixed point for the system.
This can be done simply by observing that the linearization of the system at a point $(s,1/\sqrt{V(f^{-1}(s)},0,0)$ has always $1/\sqrt{6}$ as eigenvalue, due to~\eqref{eq:mainSys-5d}.
\end{proof}
  
As done for the case $\kappa = 0$,
the study of the stability of the equilibrium points for the system~\eqref{eq:mainSys-5a}--\eqref{eq:mainSys-5d}
in $M_{-1}$ gives us the candidates to $\omega$--limit points of the system.
Since $M_{-1}$ is given by~\eqref{eq:boundaryOmega-wneg}, we have the following cases:
\begin{itemize}
    \item if $s_0 \ne \pm1$ and $(s_0,0,w_0,z_0) \in M_{-1}$ is an equilibrium point, then by~\eqref{eq:mainSys-5a} we have $w_0 = 0$, which implies $z_0^2 = 1$.
    In this case, the linearized system is
    characterized by the following matrix:
    \[
    \begin{bmatrix}
        0 & 0 & f'(f^{-1}(s_0)) & 0\\
        0 & -\frac{1}{\sqrt{6}} & 0 & 0\\
        0 & 0 & \sqrt{\frac{2}{3}} &  0\\
        0 & 0 & 0 & -\sqrt{\frac{2}{3}} 
        \end{bmatrix},
    \]
    hence the two lines given by 
    $\big\{(s,0,0,\pm1): s \in (-1,1)\big\}$
    are unstable equilibria for the system.
    \item 
    if $s_0 = \pm 1$, then 
    we can reduce to study the system 
    only with respect to $(s,w,z)$.
    By~\eqref{eq:mainSys-5d}, we have an equilibrium point if either $z_0 = 0$
    or $1-3w_0^2 - z_0^2 = 0$.
    Let us examine separately the two cases:
    \begin{itemize}
        \item $s_0 = \pm 1$ and $z_0 = 0$:
        we have considerations analogous to the case where $\kappa = 0$.
        Specifically, the physically relevant points $(1,1,0)$ and $(-1,-1,0)$ are stable if $u(1) < 1$ and $u(-1) > -1$, respectively.
        Moreover, the other two points of equilibrium, which are $(1,u(1),0)$ and $(-1,u(-1),0)$, are always unstable.
    \item $s_0 = \pm1$ and $1-3w_0^2 -z_0^2 = 0$:
        in such points,~\eqref{eq:mainSys-5c}
        reads as follows:
        \[
        \frac{\de w}{\de\tau}
        = \sqrt{\frac{2}{3}}w_0\big(1 - 3w_0u(s_0)\big).
        \]
        As a consequence, 
        such points are of equilibrium  if either $w_0 = 0$ or $3w_0u(s_0) = 1$.
        If $w_0 = 0$, then $z_0^2 = 1$
        and the linearized system is given by the following matrix:
        \[
        \begin{bmatrix}
            0 & 0 & 0 \\
            0 & \sqrt{\frac{2}{3}} & \sqrt{6}u(s_0)z_0 \\
            0 & 0 & -\sqrt{\frac{2}{3}}
        \end{bmatrix},
        \]
        so the point is an unstable equilibrium.
        On the other hand,
        if $s_0 = \pm1$, 
        $1 - 3w_0^2 + z_0^2 = 0$
        and
        $3w_0u(s_0) = 1$,
        the linearized system is characterized by the following matrix:
        \[
        \begin{bmatrix}
            w_0 \lambda_{\pm} & 0 & 0\\
            -\sqrt{6}u'(s_0)w_0^2 & \sqrt{\frac{2}{3}}(1 + z_0^2) & 
            \sqrt{\frac{2}{3}}z_0(3u(s_0) - w_0) \\
            0 & -\sqrt{6}w_0z_0 & -\sqrt{\frac{2}{3}}z_0^2 
        \end{bmatrix},
        \]
        whose eigenvalues are
        $w_0\lambda_{\pm} < 0$
        (where we assume
        that $w_0$ and $s_0$ have the same sign),
        and 
        $\frac{1}{\sqrt{6}}(1 \pm \sqrt{
        1 - 4z_0^2})$,
        so it always has an eigenvalue with positive real part.
        Therefore, also these equilibrium points are unstable.
    \end{itemize}    
\end{itemize}

Summarizing, by the above considerations we obtain that Proposition~\ref{prop:sf} holds even in this case,
namely if $u(-1)>-1$ or $u(1)<1$ then the scalar field generically diverges, meaning that $s^2 \to 1$, with a velocity such that $w^2 \to 1$
and $z^2 \to 0$.

\section{A collapsing global model}%
\label{sec:global_model_the_exterior}

The above models can be considered as matter interior of a global model obtained performing a suitable junction with an exterior spacetime, thereby obtaining models of collapsing objects, whose endstate can be investigated.

A natural choice for the exterior is
the so--called \emph{generalized Vaidya solution}, which is
the spacetime generated by a radiating object \cite{Wang:1998qx} and 
it reads as follows:
\begin{equation}\label{eq:GV}
\text ds_{\mathrm{ext}}^2=-\left(1-\frac{2M(U,R)}R\right)\,\mathrm
dU^2-2\,\mathrm
dR\,\mathrm dU + R^2\,\mathrm d\Omega^2,
\end{equation}
where $M$ is an arbitrary (positive) function.
The matching is performed along a hypersurface $\Sigma=\{r=r_b\}$. As proved in \cite{Giambo:2005se,GGM-JMP2007}, the junction conditions of Darmois--Israel between these two models require continuity of $M$ at the junction hypersurface, and the additional condition 
$\tfrac{\partial M}{\partial U}=0$, again on the junction hypersurface.

Having established a global model, our next task is to investigate whether recollapse occurs within a finite amount of comoving time.
Let us start by recalling that $\tfrac{\de t}{\de\tau}=-\tfrac1{\sqrt 6 h}=x/\sqrt 2.$ The time of collapse is then given by
\begin{equation}\label{eq:ts}
t_s=\frac{1}{\sqrt 2}\int_0^\infty x(\tau)\,\mathrm d\tau.    
\end{equation}
On the other hand, \eqref{eq:mainSys3-b} implies
\[
\frac{\de x}{\de\tau}
				= -\frac{x}{\sqrt{6}}(3w^2 - \kappa z^2),
\]
and therefore we obtain
\begin{equation}\label{eq:xtau}
x(\tau)=x_0\,\mathrm{exp}\left(\frac1{\sqrt 6}\int_0^\tau \big(\kappa z(\tilde\tau)^2-3w(\tilde\tau)^2\big)\mathrm d\tilde\tau\right).
\end{equation}
Let us consider, for example, the case $\kappa=0$. We have seen that, under suitable assumptions on $V(\phi)$, generically $w^2\to 1$.
By~\eqref{eq:xtau}, we have $x(\tau)\sim e^{-\tau\sqrt{3/2}}$ and then the integral in \eqref{eq:ts} converges, so the solution collapses in a finite amount of comoving time. 

There is also a non generic situation in the case $\kappa=0$ where $w\to u(1)<1$ (a similar argument applies to $u(-1)>-1$).
In this case, by~\eqref{eq:xtau}
we infer
$x(\tau)\sim e^{-\tau\sqrt{3/2}u(1)^2}$,
so again, when $u(1)>0$, the integral converges and the singularity forms in a finite amount of time. Observe that, in case $w\to u(1)=0$, one can conceive situations where 
\[
\int_0^\infty w(\tau)^2\,\mathrm d\tau
\]
diverges.
In this cases, the integral \eqref{eq:ts} may also diverge, indicating a solution that indefinitely collapses, remaining regular eternally.

\subsection{Horizon formation}
\label{sec:horizon_formation}
In addition to the above analysis, one can examine the causal character of the solution when approaching (finitely or infinitely) the singularity. For this purpose, one can use the results from \cite{Giambo:2005se}, where it is shown that if $\dot a(t)$ remains bounded when approaching the singularity, the boundary of the interior can be chosen so small that the apparent horizon does not form (see Figure \ref{fig:horizon}).

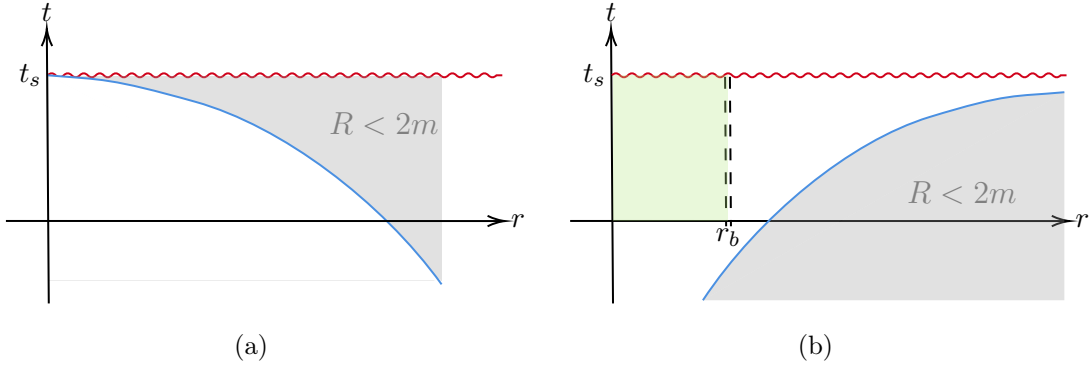
\begin{figure}
    \centering
\subfigure[\ ]
{
\tikzset{every picture/.style={line width=0.75pt}} 
\begin{tikzpicture}[x=0.75pt,y=0.75pt,yscale=-0.8,xscale=0.8]

\draw [color={rgb, 255:red, 208; green, 2; blue, 27 }  ,draw opacity=1 ]   (60.72,119.56) .. controls (62.39,117.89) and (64.05,117.89) .. (65.72,119.56) .. controls (67.39,121.23) and (69.05,121.23) .. (70.72,119.56) .. controls (72.39,117.89) and (74.05,117.89) .. (75.72,119.56) .. controls (77.39,121.23) and (79.05,121.23) .. (80.72,119.56) .. controls (82.39,117.89) and (84.05,117.89) .. (85.72,119.56) .. controls (87.39,121.23) and (89.05,121.23) .. (90.72,119.56) .. controls (92.39,117.89) and (94.05,117.89) .. (95.72,119.56) .. controls (97.39,121.23) and (99.05,121.23) .. (100.72,119.56) .. controls (102.39,117.89) and (104.05,117.89) .. (105.72,119.56) .. controls (107.39,121.23) and (109.05,121.23) .. (110.72,119.56) .. controls (112.39,117.89) and (114.05,117.89) .. (115.72,119.56) .. controls (117.39,121.23) and (119.05,121.23) .. (120.72,119.56) .. controls (122.39,117.89) and (124.05,117.89) .. (125.72,119.56) .. controls (127.39,121.23) and (129.05,121.23) .. (130.72,119.56) .. controls (132.39,117.89) and (134.05,117.89) .. (135.72,119.56) .. controls (137.39,121.23) and (139.05,121.23) .. (140.72,119.56) .. controls (142.39,117.89) and (144.05,117.89) .. (145.72,119.56) .. controls (147.39,121.23) and (149.05,121.23) .. (150.72,119.56) .. controls (152.39,117.89) and (154.05,117.89) .. (155.72,119.56) .. controls (157.39,121.23) and (159.05,121.23) .. (160.72,119.56) .. controls (162.39,117.89) and (164.05,117.89) .. (165.72,119.56) .. controls (167.39,121.23) and (169.05,121.23) .. (170.72,119.56) .. controls (172.39,117.89) and (174.05,117.89) .. (175.72,119.56) .. controls (177.39,121.23) and (179.05,121.23) .. (180.72,119.56) .. controls (182.39,117.89) and (184.05,117.89) .. (185.72,119.56) .. controls (187.39,121.23) and (189.05,121.23) .. (190.72,119.56) .. controls (192.39,117.89) and (194.05,117.89) .. (195.72,119.56) .. controls (197.39,121.23) and (199.05,121.23) .. (200.72,119.56) .. controls (202.39,117.89) and (204.05,117.89) .. (205.72,119.56) .. controls (207.39,121.23) and (209.05,121.23) .. (210.72,119.56) .. controls (212.39,117.89) and (214.05,117.89) .. (215.72,119.56) .. controls (217.39,121.23) and (219.05,121.23) .. (220.72,119.56) .. controls (222.39,117.89) and (224.05,117.89) .. (225.72,119.56) .. controls (227.39,121.23) and (229.05,121.23) .. (230.72,119.56) .. controls (232.39,117.89) and (234.05,117.89) .. (235.72,119.56) .. controls (237.39,121.23) and (239.05,121.23) .. (240.72,119.56) .. controls (242.39,117.89) and (244.05,117.89) .. (245.72,119.56) .. controls (247.39,121.23) and (249.05,121.23) .. (250.72,119.56) .. controls (252.39,117.89) and (254.05,117.89) .. (255.72,119.56) .. controls (257.39,121.23) and (259.05,121.23) .. (260.72,119.56) .. controls (262.39,117.89) and (264.05,117.89) .. (265.72,119.56) .. controls (267.39,121.23) and (269.05,121.23) .. (270.72,119.56) .. controls (272.39,117.89) and (274.05,117.89) .. (275.72,119.56) .. controls (277.39,121.23) and (279.05,121.23) .. (280.72,119.56) .. controls (282.39,117.89) and (284.05,117.89) .. (285.72,119.56) .. controls (287.39,121.23) and (289.05,121.23) .. (290.72,119.56) .. controls (292.39,117.89) and (294.05,117.89) .. (295.72,119.56) .. controls (297.39,121.23) and (299.05,121.23) .. (300.72,119.56) .. controls (302.39,117.89) and (304.05,117.89) .. (305.72,119.56) .. controls (307.39,121.23) and (309.05,121.23) .. (310.72,119.56) .. controls (312.39,117.89) and (314.05,117.89) .. (315.72,119.56) .. controls (317.39,121.23) and (319.05,121.23) .. (320.72,119.56) .. controls (322.39,117.89) and (324.05,117.89) .. (325.72,119.56) .. controls (327.39,121.23) and (329.05,121.23) .. (330.72,119.56) .. controls (332.39,117.89) and (334.05,117.89) .. (335.72,119.56) .. controls (337.39,121.23) and (339.05,121.23) .. (340.72,119.56) -- (344.38,119.56) -- (344.38,119.56) ;
\draw  [draw opacity=0][fill={rgb, 255:red, 155; green, 155; blue, 155 }  ,fill opacity=0.3 ] (60.72,120) -- (306.5,120) -- (306.5,248.5) -- (60.72,248.5) -- cycle ;
\draw  [draw opacity=0][fill={rgb, 255:red, 255; green, 255; blue, 255 }  ,fill opacity=1 ] (306.5,248.5) -- (61.5,119.56) -- (61.5,248.5) -- cycle ;
\draw [color={rgb, 255:red, 74; green, 144; blue, 226 }  ,draw opacity=1 ][fill={rgb, 255:red, 255; green, 255; blue, 255 }  ,fill opacity=1 ]   (306.5,251) .. controls (268.86,198.52) and (206.56,151.38) .. (152.94,136.24) .. controls (99.31,121.09) and (94.91,122.37) .. (60.72,119.56) ;
\draw    (34.66,211.14) -- (344.12,211.2) ;
\draw [shift={(346.12,211.2)}, rotate = 180.01] [color={rgb, 255:red, 0; green, 0; blue, 0 }  ][line width=0.75]    (10.93,-3.29) .. controls (6.95,-1.4) and (3.31,-0.3) .. (0,0) .. controls (3.31,0.3) and (6.95,1.4) .. (10.93,3.29)   ;
\draw    (61.4,263) -- (60.18,92.25) ;
\draw [shift={(60.17,90.25)}, rotate = 89.59] [color={rgb, 255:red, 0; green, 0; blue, 0 }  ][line width=0.75]    (10.93,-3.29) .. controls (6.95,-1.4) and (3.31,-0.3) .. (0,0) .. controls (3.31,0.3) and (6.95,1.4) .. (10.93,3.29)   ;

\draw (234.5,141) node [anchor=north west][inner sep=0.75pt]  [color={rgb, 255:red, 128; green, 128; blue, 128 }  ,opacity=1 ] [align=left] {$\displaystyle R< 2m$};
\draw (58.72,119.56) node [anchor=east] [inner sep=0.75pt]  [font=\small] [align=left] {$\displaystyle t_{s}$};
\draw (348.12,211.2) node [anchor=west] [inner sep=0.75pt]  [font=\small] [align=left] {$\displaystyle r$};
\draw (60.17,87.25) node [anchor=south] [inner sep=0.75pt]  [font=\small] [align=left] {$\displaystyle t$};

\end{tikzpicture}
}
\subfigure[\ ]
{
\tikzset{every picture/.style={line width=0.75pt}} 
\begin{tikzpicture}[x=0.75pt,y=0.75pt,yscale=-0.8,xscale=0.8]

\draw    (34.66,211.14) -- (344.12,211.2) ;
\draw [shift={(346.12,211.2)}, rotate = 180.01] [color={rgb, 255:red, 0; green, 0; blue, 0 }  ][line width=0.75]    (10.93,-3.29) .. controls (6.95,-1.4) and (3.31,-0.3) .. (0,0) .. controls (3.31,0.3) and (6.95,1.4) .. (10.93,3.29)   ;
\draw    (61.4,263) -- (60.18,92.25) ;
\draw [shift={(60.17,90.25)}, rotate = 89.59] [color={rgb, 255:red, 0; green, 0; blue, 0 }  ][line width=0.75]    (10.93,-3.29) .. controls (6.95,-1.4) and (3.31,-0.3) .. (0,0) .. controls (3.31,0.3) and (6.95,1.4) .. (10.93,3.29)   ;
\draw [color={rgb, 255:red, 208; green, 2; blue, 27 }  ,draw opacity=1 ]   (60.72,119.56) .. controls (62.39,117.89) and (64.05,117.89) .. (65.72,119.56) .. controls (67.39,121.23) and (69.05,121.23) .. (70.72,119.56) .. controls (72.39,117.89) and (74.05,117.89) .. (75.72,119.56) .. controls (77.39,121.23) and (79.05,121.23) .. (80.72,119.56) .. controls (82.39,117.89) and (84.05,117.89) .. (85.72,119.56) .. controls (87.39,121.23) and (89.05,121.23) .. (90.72,119.56) .. controls (92.39,117.89) and (94.05,117.89) .. (95.72,119.56) .. controls (97.39,121.23) and (99.05,121.23) .. (100.72,119.56) .. controls (102.39,117.89) and (104.05,117.89) .. (105.72,119.56) .. controls (107.39,121.23) and (109.05,121.23) .. (110.72,119.56) .. controls (112.39,117.89) and (114.05,117.89) .. (115.72,119.56) .. controls (117.39,121.23) and (119.05,121.23) .. (120.72,119.56) .. controls (122.39,117.89) and (124.05,117.89) .. (125.72,119.56) .. controls (127.39,121.23) and (129.05,121.23) .. (130.72,119.56) .. controls (132.39,117.89) and (134.05,117.89) .. (135.72,119.56) .. controls (137.39,121.23) and (139.05,121.23) .. (140.72,119.56) .. controls (142.39,117.89) and (144.05,117.89) .. (145.72,119.56) .. controls (147.39,121.23) and (149.05,121.23) .. (150.72,119.56) .. controls (152.39,117.89) and (154.05,117.89) .. (155.72,119.56) .. controls (157.39,121.23) and (159.05,121.23) .. (160.72,119.56) .. controls (162.39,117.89) and (164.05,117.89) .. (165.72,119.56) .. controls (167.39,121.23) and (169.05,121.23) .. (170.72,119.56) .. controls (172.39,117.89) and (174.05,117.89) .. (175.72,119.56) .. controls (177.39,121.23) and (179.05,121.23) .. (180.72,119.56) .. controls (182.39,117.89) and (184.05,117.89) .. (185.72,119.56) .. controls (187.39,121.23) and (189.05,121.23) .. (190.72,119.56) .. controls (192.39,117.89) and (194.05,117.89) .. (195.72,119.56) .. controls (197.39,121.23) and (199.05,121.23) .. (200.72,119.56) .. controls (202.39,117.89) and (204.05,117.89) .. (205.72,119.56) .. controls (207.39,121.23) and (209.05,121.23) .. (210.72,119.56) .. controls (212.39,117.89) and (214.05,117.89) .. (215.72,119.56) .. controls (217.39,121.23) and (219.05,121.23) .. (220.72,119.56) .. controls (222.39,117.89) and (224.05,117.89) .. (225.72,119.56) .. controls (227.39,121.23) and (229.05,121.23) .. (230.72,119.56) .. controls (232.39,117.89) and (234.05,117.89) .. (235.72,119.56) .. controls (237.39,121.23) and (239.05,121.23) .. (240.72,119.56) .. controls (242.39,117.89) and (244.05,117.89) .. (245.72,119.56) .. controls (247.39,121.23) and (249.05,121.23) .. (250.72,119.56) .. controls (252.39,117.89) and (254.05,117.89) .. (255.72,119.56) .. controls (257.39,121.23) and (259.05,121.23) .. (260.72,119.56) .. controls (262.39,117.89) and (264.05,117.89) .. (265.72,119.56) .. controls (267.39,121.23) and (269.05,121.23) .. (270.72,119.56) .. controls (272.39,117.89) and (274.05,117.89) .. (275.72,119.56) .. controls (277.39,121.23) and (279.05,121.23) .. (280.72,119.56) .. controls (282.39,117.89) and (284.05,117.89) .. (285.72,119.56) .. controls (287.39,121.23) and (289.05,121.23) .. (290.72,119.56) .. controls (292.39,117.89) and (294.05,117.89) .. (295.72,119.56) .. controls (297.39,121.23) and (299.05,121.23) .. (300.72,119.56) .. controls (302.39,117.89) and (304.05,117.89) .. (305.72,119.56) .. controls (307.39,121.23) and (309.05,121.23) .. (310.72,119.56) .. controls (312.39,117.89) and (314.05,117.89) .. (315.72,119.56) .. controls (317.39,121.23) and (319.05,121.23) .. (320.72,119.56) .. controls (322.39,117.89) and (324.05,117.89) .. (325.72,119.56) .. controls (327.39,121.23) and (329.05,121.23) .. (330.72,119.56) .. controls (332.39,117.89) and (334.05,117.89) .. (335.72,119.56) .. controls (337.39,121.23) and (339.05,121.23) .. (340.72,119.56) -- (344.38,119.56) -- (344.38,119.56) ;
\draw [color={rgb, 255:red, 74; green, 144; blue, 226 }  ,draw opacity=1 ][fill={rgb, 255:red, 155; green, 155; blue, 155 }  ,fill opacity=0.3 ]   (117.5,261) .. controls (152.15,208.32) and (209.5,161) .. (258.86,145.8) .. controls (308.22,130.59) and (311.56,133.01) .. (343.03,130.19) ;
\draw  [draw opacity=0][fill={rgb, 255:red, 155; green, 155; blue, 155 }  ,fill opacity=0.3 ] (117.5,261) -- (343.03,130.19) -- (343.03,261) -- cycle ;
\draw  [dash pattern={on 4.5pt off 4.5pt}]  (134.34,120.51) -- (134.99,214.48)(131.34,120.53) -- (131.99,214.5) ;
\draw  [draw opacity=0][fill={rgb, 255:red, 184; green, 233; blue, 134 }  ,fill opacity=0.3 ] (61.5,119.56) -- (133.16,119.56) -- (133.16,211.29) -- (61.5,211.29) -- cycle ;

\draw (243,184) node [anchor=north west][inner sep=0.75pt]  [color={rgb, 255:red, 128; green, 128; blue, 128 }  ,opacity=1 ] [align=left] {$\displaystyle R< 2m$};
\draw (58.72,119.56) node [anchor=east] [inner sep=0.75pt]  [font=\small] [align=left] {$\displaystyle t_{s}$};
\draw (133.16,214.29) node [anchor=north] [inner sep=0.75pt]  [font=\small] [align=left] {$\displaystyle r_{b}$};
\draw (348.12,211.2) node [anchor=west] [inner sep=0.75pt]  [font=\small] [align=left] {$\displaystyle r$};
\draw (60.17,87.25) node [anchor=south] [inner sep=0.75pt]  [font=\small] [align=left] {$\displaystyle t$};
\end{tikzpicture}
}
    \caption{(a): if $\dot a(t)$ is unbounded, at the center of the model the horizon (blue line) forms at the same comoving time of the singularity (red line). (b): when $\dot a(t)$ is bounded in the approach to the singularity, one can choose a sufficiently small neighborhood of the center such that $R>2m $ and perform a junction of the interior (the green region) with an exterior spacetime.}
    \label{fig:horizon}
\end{figure}

This fact also holds when $\kappa=-1$.
By applying the definition of the Misner-Sharp mass $1-2m/R=g(\nabla R,\nabla R)$ to the the metric \eqref{eq:RWmodel} we obtain the following implications:
\begin{equation}\label{eq:horizon}
    R>2m \quad\Longleftrightarrow\quad 1+\dot a(t)  r - \kappa r^2>0\quad\Longleftrightarrow\quad -\dot a(t)<\frac{1-\kappa r^2}{r}
\end{equation}
(recall that $\dot a(t)$ is the derivative w.r.t. comoving time $t$ and it is negative in our model). 

In our situation $\dot a(t)$ is related, recalling \eqref{eq:def-xwz}, to the unknown $z(\tau)$, whose behavior is given by equation
\eqref{eq:mainSys3-d}, which we recall reads as follows:
\begin{equation*}
\frac{\de z}{\de\tau}
		= \frac{z}{\sqrt{6}}
		\left(1 - 3w^2 + \kappa z^2 \right).
\end{equation*}
Again, let us briefly review the case $\kappa=0$. The generic situation is given by $w^2\to 1$, in such a way that $z\sim e^{-\tau\sqrt{2/3}}$. This means that $z\to 0$ as $\tau\to+\infty$ and then $\dot a$ diverges, resulting in the formation of the horizon and then in a black hole. Viceversa, in the non generic situation $w\to u(1)<1$, we get 
$$z\sim e^{(1-3 u(1)^2)\tau}.
$$ 
Therefore, if 
$$
1-3 u(1)^2>0, 
$$
(for example if $u(1)=0$) then $z$ diverges and then $\dot a$ is bounded in the approach to the singularity. Because of this, the model can be built in such a way that the horizon does not form, and the singularity is naked. 

We summarize the above findings in the following result.
\begin{proposition}\label{prop:hor}
    If the limit
$$\lim_{\tau\to+\infty}3w(\tau)^2-\kappa z(\tau)^2=:\ell
$$
exists and it is positive, then $t_s\in\mathbb R$, i.e., the solution collapses in a finite amount of comoving time.

Moreover, if $\ell>1$ then a horizon forms, and the model exhibits a black hole. Otherwise, if $\ell<1$, then choosing the boundary of the internal solution sufficiently small, the horizon does not form, resulting in a naked singularity.
\end{proposition}

In view of the analysis of the qualitative behavior performed in Sections \ref{sec:the_flat_case_kappa_0_} and \ref{sec:collapse_of_the_open_spatial_topology_kappa_1_}, we can conclude with the following theorem.
\begin{teo}\label{thm:endstate}
    If $u(-1)>-1$ or $u(1)<1$, then the collapsing scalar field with open spatial topology $($i.e., $\kappa=0$ or $\kappa = -1$$)$ generically collapses in a finite amount of comoving time forming a black hole. 
\end{teo}

\section{Horizon formation genericity}%
\label{sec:horizon_formation_genericity}
In the results stated above,
genericity is meant with respect
to the initial data of the dynamical system:
even if there may exist solutions that don't develop a horizon, 
and therefore they collapse to a naked singularity,
a small perturbation of its initial data will restore black hole formation.

On the other side, one can conceive the same model but prescribing the behaviour of some physical quantity, e.g., the energy density,
with respect to the scale factor, and see
whether the endstate associated to a given prescription is stable w.r.t. perturbation of the prescription itself.
This analysis has been discussed --
at least for the flat case $\kappa = 0$ --
in the paper \cite{Giambo:2005se}
where a function $\psi(a)$ is prescribed such that
\begin{equation}
	\label{eq:psia}
	\dot{a} = - \psi(a).
\end{equation}
In the general case, field equations \eqref{eq:Einstein-0}--\eqref{eq:Einstein-i} yield the following expressions
for $\displaystyle \left(\frac{\de\phi}{\de a}\right)^2$
and $V = V(a)$
(see \cite{Giambo:2005se,Giambo2024}):
\begin{subequations}
	\begin{align}
		\left(\frac{\de\phi}{\de a}\right)^2
		&= \frac{2}{a^2}\left[1 - a\frac{\psi'(a)}{\psi(a)}
			+ \frac{\kappa}{\psi^2(a)}
		\right],
		\label{eq:dphida2}\\
		V(a) &= 
		2 \left(\frac{\psi(a)}{a}\right)^2
		\left[1 + \frac{1}{2}\frac{a\psi'(a)}{\psi(a)}
			+ \frac{\kappa}{\psi^2(a)}
		\right].
		\label{eq:Va}
	\end{align}
\end{subequations}
Within the above framework,
generical examples of collapse can conceived, leading to naked singularities,
where genericity now is intended w.r.t. perturbation of some parameter entering the relation~\eqref{eq:psia}. 
Let us review a few examples already known in literature,
verifying that they correspond to non generic situation (with respect to perturbations of initial data).

\begin{example}\label{ex:Goswami}
In the paper \cite{Goswami:2005fu}, the loop quantum gravity modification of a collapsing scalar field in the case $\kappa = 0$ is considered, with the additional ansatz that the the energy density $\rho$
satisfies the following identity:
\begin{equation}
    \label{eq:rho-example1}
    \rho = \frac{1}{a^{\nu}},
\end{equation}
where $\nu$ is a positive parameter.
It is shown that the additional condition $\nu<2$ -- that is an open condition, hence stable with respect to perturbation of $\nu$ -- produces a solution collapsing in a finite amount of comoving time, where the horizon does not form and then the singularity is naked. Let us reconsider this example in the framework developed in the previous sections. 

Recalling \eqref{eq:Einstein-0} we have
	\[
		\rho = \frac{1}{2}\dot{\phi}^2 + V(\phi)
		= 3\left(\frac{\dot{a}}{a}\right)^2,
	\]
	which implies that 
	\[
		\dot{a} = -\frac{1}{\sqrt{3}}a^{1-\frac{\nu}{2}},
	\]
	and then the function $\psi(a)$ in~\eqref{eq:psia} is given by:
	\[
		\psi(a) = \frac{1}{\sqrt{3}}a^{1-\frac{\nu}{2}}.
	\]
	By~\eqref{eq:dphida2}, we have
	\[
	\left(\frac{\de\phi}{\de a}\right)^2
	= \frac{2}{a^2}\left[-1 + 1 + \frac{\nu}{2}\right]
	= \frac{\nu}{a^2},
	\]
	and then $\frac{\de\phi}{\de a} = \pm \frac{\sqrt{\nu}}{a}$.
	Taking for instance the positive root, we have
	$\phi = \ln a^{\sqrt{\nu}}$,
	i.e., $a = e^{\phi/\sqrt{\nu}}$,
	where the iteration constant has been chosen for simplicity. 
	Then the collapse happens for $\phi \to -\infty$.
	Using~\eqref{eq:Va}
	we have
	\[
		V(a) = \frac{2}{3}a^{-\nu}
		\left(1 + \frac{1}{2}
			\left(1 - \frac{\nu}{2}\right)
		\right)
		= \left(1 - \frac{\nu}{6}\right)\frac{1}{a^{\nu}}.
	\]
	Combining this formula with the previous result, we get
	\begin{equation}
		\label{eq:Vphi-example1}
		V(\phi) = \left(1 - \frac{\nu}{6}\right)e^{\sqrt{\nu}|\phi|},
	\end{equation}
	where we have introduced the absolute value
	to get rid of the ambiguity coming from the square root
	extraction seen above.
	From~\eqref{eq:Vphi-example1} we derive
	\begin{equation}
		\label{eq:example1-uphi}
		u(\phi) = \sqrt{\frac{\nu}{6}}\sgn(\phi),
	\end{equation}
	and then, by using the notation 
	introduced in Section \ref{sec:the_flat_case_kappa_0_},
	we have $u(-1) = -\sqrt{\nu/6}$
	and $u(1) = \sqrt{\nu/6}$.

	On the other hand, we have
	\[
		w = - \frac{\dot{\phi}}{\sqrt{6}\, h}
		= - \frac{\de \phi}{\de a}\dot{a} \frac{a}{\sqrt{6}\dot{a}}
		= - \frac{a}{\sqrt{6}}\frac{\sqrt{\nu}}{a}
		= -\sqrt{\frac{\nu}{6}}.
	\]
	Then, the prescription~\eqref{eq:rho-example1}
	corresponds to the situation where 
	$w \to u(-1)$,
	i.e., one of the non generic choices discussed in 
	Theorem~\ref{theorem:x0}.
 
	In view of the results of Section \ref{sec:global_model_the_exterior}, this solution leads to a naked
	singularity when $1 - 3u^2(-1) > 0$,
	i.e., $\nu < 2$,
	as remarked in \cite{Goswami:2005fu}.
\end{example}

\begin{example}\label{ex:Baier}
   The paper \cite{Baier:2014ita} presents a modification of the situation presented in the previous example, considering the Einstein field equations explicitly adding a (negative) cosmological constant. As a matter of fact, due to the form of the energy momentum tensor \eqref{eq:momentumTensor}, the presence of $\Lambda$ adds a constant to the potential function $V(\phi)$.
   With this in mind, one has to expect a  situation similar to the one that emerges from the previous example. 

   In \cite{Baier:2014ita}, it is prescribed the following functional dependence of the energy density with respect to the scale factor:
   \begin{equation}
   \label{eq:Baier}
   \rho(a)=3 a^{2 \beta -2},
\end{equation}
with $\beta<1$ so that $\rho$ diverges as $a\to 0^+$. This amounts to prescribe the function $\psi(a)$ in \eqref{eq:psia} as follows:
\begin{equation}\label{eq:rho-Baier}
\psi(a)=\sqrt{a^{2 \beta }-\frac{a^2}{l^2}},
\end{equation}
where the cosmological constant $\Lambda$ is set equal to $-\tfrac3{l^2}$, using the notation of \cite{Baier:2014ita}. 
Since $\beta<1$, $\psi(a)\simeq a^\beta$ as $a\to 0^+$, and therefore $1/\psi(a)$ is integrable in a right neighborhood of $a=0$, the singularity forms in a finite amount of comoving time.
Moreover, assuming also $\beta>0$, $\psi(a)$ is bounded and therefore a naked singularity can be built from this collapse, exactly as in the previous example. Again, one can reconstruct back the potential $V(\phi)$, whose complete expression is given in \cite[(23)]{Baier:2014ita} and has an exponential growth. 
In this case, we have (embodying $\Lambda$ in $V$) 
$$
\dot\phi^2=
2 (1-\beta) a^{2 \beta -2},
\qquad V=(\beta +2) a^{2 \beta -2}-\frac{3}{l^2},
$$
and then we can find $w$ as a function of $a$, satifying the following relation:
$$
w^2=\frac{(\beta -1) l^2 a^{2 \beta }}{3 \left(a^2-l^2 a^{2 \beta }\right)}\xrightarrow{a\to 0}\frac{1-\beta }{3}.
$$
At the same time, we can compute $u^2(1)$,
obtaining 
\begin{multline*}
u^2(1)=\lim_{a\to 0^+}\left(\frac{1}{\sqrt{6}}\frac{\de V}{\de a}\frac{(-\psi(a))}{\dot\phi}\right)^2\\
=\lim_{a\to 0^+}\frac{(1 -\beta) (\beta +2)^2 l^2 a^{2 \beta } \left(l^2 a^{2 \beta }-a^2\right)}{3 \left((\beta +2) l^2 a^{2 \beta }-3 a^2\right)^2}=\frac{1-\beta }{3}.
\end{multline*}
Since these two limits are equal, we deduce that $w\to u(1)<1$, that is once again one of the non generic situations occurring in the $\kappa=0$ case. We deduce as well that the naked singularity arising from the functional dependence \eqref{eq:rho-Baier} is nongeneric with respect to perturbation of the initial conditions. 
\end{example}

\begin{example}\label{ex:Mosani}
	In \cite{Mosani:2021czj} a collapsing scalar field in the flat case ($\kappa=0$) is considered, where it is prescribed
    the following behaviour:
	\begin{equation}
		\label{eq:example2-rho}
		\rho(a) = 64\lambda(\gamma - \log a)^2,
	\end{equation}
	where $\lambda$ and $\gamma$ are positive parameters, therefore again subjected to open -- hence stable -- conditions. The resulting collapse is shown not to form a singularity in a finite amount of time, since $a(t)$ vanishes as $t\to+\infty$. Again, let us discuss this example in view of the results stated in the previous sections. 

	Arguing as in Examples \ref{ex:Goswami} and \ref{ex:Baier}, we have
	\[
		\psi(a) = 8a\sqrt{\frac{\lambda}{3}}(\gamma - \log a),
	\]
	\[
	\left(\frac{\de\phi}{\de a}\right)^2
	= \frac{2}{a^2(\gamma - \log a)}
	\implies \phi(a) = -2\sqrt{2(\gamma - \log a)}
	\implies a = e^{\gamma - \frac{\phi^2}{8}},
	\]
	from which we get
	\[
		V(a) = \frac{64}{3}
		\lambda(3\gamma - 1 - 3\log a)
		(\gamma - \log a).
	\]
	Combining these equations with the previous result yields to the expression:
	\begin{equation}
		\label{eq:example2-Vphi}
		V(\phi) = \lambda\left(\phi^4-\frac{8}{3}\phi^2\right),
	\end{equation}
	that produces $u(-1) = u(1) = 0$.
        On the other hand, we have
	\[
		w = - \frac{a}{\sqrt{6}}\frac{\de\phi}{\de a}
		= - \frac{a}{\sqrt{6}}\frac{\sqrt{2}}{a\sqrt{\gamma - \log a}}
		= - \frac{1}{\sqrt{3(\gamma - \log a)}}
		\xrightarrow[]{a \to 0} 0,
	\]
	once again corresponding to one of the non generic
	choices described in Theorem~\ref{theorem:x0}, where the apparent horizon does not form.

We observe that in this case we have, from \eqref{eq:mainSys4-b} and the above expression for $w$, that
$$
\frac{\de x}{\de \tau}=-\frac1{\sqrt 6}\frac{x}{\gamma-\log a}\Rightarrow\frac{\de x}{\de a}=\frac{x}{a(\gamma-\log a)}.
$$
Hence, $x(a)=\tfrac{c}{\gamma-\log a}$,
where $c$ is an integration constant, and since 
$$
\int_0^\infty x\,\mathrm d\tau=-\sqrt{6}c\int_0^{a_0}\frac{\mathrm da}{a(\gamma-\log a)}
$$
diverges, then the solution indefinitely collapses and remains regular, as remarked at the beginning of the present example. 
 
\end{example}

\begin{example}\label{ex:3} In the wake of the previous examples, we can build a similar model when $\kappa=-1$, prescribing 
\begin{equation}\label{eq:example3-rho}
\rho(a)=\frac{3-c\, a^{2 \gamma ^2}}{a^2},
\end{equation}
with $c,\gamma\ne 0$. Arguing as before, and recalling \eqref{eq:Einstein-0} we can find that 
$$
\dot a=-\sqrt{\frac{1}{3} a^2 \rho (a)+1}=
-\sqrt{2-\frac{1}{3} c\, a^{2 \gamma ^2}}.
$$
We observe that the function
$$
a \mapsto -\frac{1}{\sqrt{2-\frac{1}{3} c\, a^{2 \gamma ^2}}}
$$
is regular, hence integrable in a right neighborhood of $a=0$, and then the singularity is reached in a finite amount of comoving time. Moreover, $\dot a$ is bounded as $a\to 0^+$, which means that a model without horizon can be built as well.
However, denoting by $w_0$ and $z_0$ the limits of $w(\tau)$ and $z(\tau)$, we can compute them using the limit of the scale factor as $a\to 0^+$:
\begin{multline*}
1-3w_0^2-z_0^2=\lim_{a\to 0^+}\left[1-
-3 \left(-\frac{a \dot\phi}{\sqrt{6} (-\psi (a))}\right)^2-\left(\frac{1}{\psi (a)}\right)^2
\right]\\
=\lim_{a\to 0^+}\gamma ^2 \left(1+\frac{6}{-6+c\, a^{2 \gamma ^2}}\right)=0,
\end{multline*}
which is again one of the nongeneric situations described in Section \ref{sec:collapse_of_the_open_spatial_topology_kappa_1_}. 

Observe that, once again, we can reconstruct back the potential $V(\phi)$ solving \eqref{eq:dphida2}, inverting $a$ as a function of $\phi$ and plugging it back in \eqref{eq:Va}. In this case the  solution cannot be written in explicit form, but  Figure \ref{fig:V} represents the behavior of $V$ for several choices of $\gamma$. 

\begin{figure}
\centering
\includegraphics[width=0.9\linewidth]{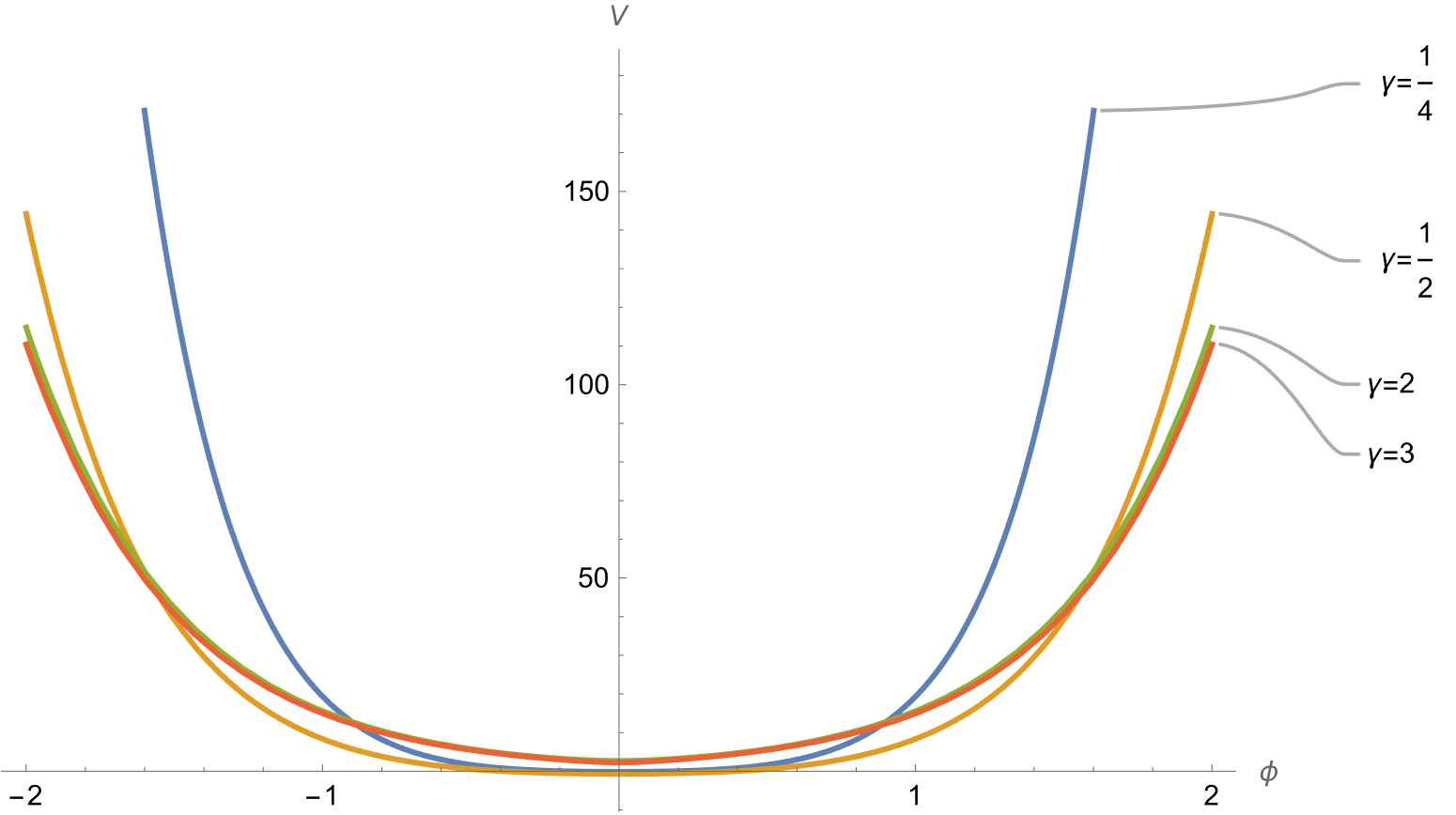}
    \caption{$V(\phi)$ for several choices of $\gamma$, see Example \ref{ex:3}. The constant $c$ in \eqref{eq:example3-rho} has been fine-tuned requiring continuity of $V'(\phi)$ in $\phi=0$.  }
    \label{fig:V}
\end{figure}

\end{example}

\section{Conclusions}\label{sec:outro}

It is well known that research on scalar fields in General Relativity has primarily developed within a cosmological framework, with significant attention given to $F(R)$ theories. Consequently, the study of scalar field evolution in a Robertson-Walker (RW) metric has predominantly focused on \textit{expanding} cosmologies (see, e.g., \cite{Harko:2013gha, Leon:2020pfy, Leon:2020ovw}). However, as outlined in the Introduction, the study of (inhomogeneous) scalar fields has also served as crucial toy models for investigating relativistic gravitational collapse \cite{Giambo2024}. Consequently, in recent years, there has been a growing interest in examining the behavior of \textit{collapsing} scalar fields within homogeneous settings.
This interest extends to alternative gravitational theories or classical relativistic contexts involving matter sources in addition to the scalar field itself (see, e.g., \cite{Giambo:2008sa, Tavakoli:2015llr, Chakrabarti:2020hrx}).

In this paper, we have reviewed the main findings concerning the collapse process, with an attempt in some cases to provide proofs of these theorems more closely aligned with dynamical systems theory, akin to pioneering works in this field \cite{Foster:1998sk, Miritzis:2003eu, Miritzis:2003ym}. For instance, Theorem \ref{theorem:x0}, pertaining to the flat case $\kappa=0$, was treated in \cite[Theorem 3.6]{GGM-JMP2007} with an \textit{ad-hoc} argument utilizing the Implicit Function Theorem. The consideration of closed spatial topologies has been approached from various perspectives, including alternative theories of gravity (see, e.g., \cite{Astashenok:2018bol, Ziaie:2021yox, Koutsoumbas:2015ekk}).

We demonstrated that similar techniques can be successfully applied, with necessary adjustments, to cases with open spatial topologies ($\kappa=1$), as shown in Theorem \ref{theorem:x0-wneg}. 
This is primarily due to Remark \ref{rem:Rem2}, which we discussed at the outset, highlighting the monotonicity of the Hubble function $h(t)$.
On the contrary, in the case of closed spatial topology ($\kappa=1$), an increasing collapse may lead to competition between the terms on the right-hand side of \eqref{eq:h-ODE}, potentially resulting in a reversal of the evolution towards an expansion phase. Moreover, the compactness of the set where the trajectories of the model are analyzed (see \eqref{eq:Constraint}) offers a reliable basis for cases where $\kappa\leq0$, but this compactness is not guaranteed when dealing with $\kappa=1$.

In Section \ref{sec:horizon_formation_genericity}, we have further analyzed the potential endstates of collapse for specific models, including notable cases documented in the literature in the flat case
\cite{Goswami:2005fu, Baier:2014ita, Mosani:2021czj},
with and without cosmological constant, and presenting an additional example for $\kappa=-1$. This discussion is closely tied to the genericity aspects of naked singularity occurrences in these models, which in turn relates to the well-known Penrose cosmic censorship conjecture. The determination of genericity hinges on the set where parameter perturbations are applied. When considering genericity concerning perturbations of initial conditions governing the dynamical systems associated with these models, cosmic censorship is restored, as evidenced by the results presented in Section \ref{sec:horizon_formation}, particularly in Theorem \ref{thm:endstate}. Nevertheless, it is important to note, as emphasized in the Introduction, that these non-generic naked singularities are gravitationally strong \cite{Guo:2020ked}.

The results presented in this paper rely on general assumptions about the potential function $V(\phi)$, particularly its behavior at infinity.
Specifically, we assume an at-most-exponential growth, as outlined in Assumption \ref{as:V}. Finally, it is noteworthy that the introduction of a cosmological constant, whether negative \cite{Baier:2014ita} or positive, does not fundamentally alter the qualitative framework.


\printbibliography

\end{document}